\documentclass[12pt,a4paper]{article} %
\usepackage{natbib}
\usepackage{amssymb}
\usepackage{graphicx,amsmath} 
\usepackage{graphics}
\usepackage{amsthm}
\usepackage{subcaption}
\usepackage{enumitem}
\usepackage[colorlinks=true,linkcolor=blue, urlcolor=blue, citecolor=blue, breaklinks]{hyperref}

\newtheorem{proposition}{Proposition}

\newtheorem{result}{Result}
\newtheorem{hypo}{Hypothesis}
\newtheorem{definition}{Definition}

\newcommand{\RR}{\mathbb{R}}

\newcommand{\EE}{\mathbb{E}}

\usepackage[misc]{ifsym}
\usepackage[flushleft]{threeparttable}

\begin{document}
\title{The Strength of Absent Ties:\\
 Social Integration via Online Dating}
\author{Philipp Hergovich and Josu\'e Ortega\footnote{Corresponding author: \href{mailto:ortega@zew.de}{ortega@zew.de} (\Letter).} \thanks{Hergovich: University of Vienna. Ortega: Center for European Economic Research (ZEW). We are particularly indebted to Dilip Ravindran for his many valuable comments. We also acknowledge helpful written feedback from So Yoon Ahn, 	Andrew Clausen, Melvyn Coles, Karol Mazur, David Meyer, Patrick Harless, Misha Lavrov, Franz Ostrizek, Yasin Ozcan, Gina Potarca, Reuben Thomas, MSE ``quasi'' and audiences at the Universities of Columbia and Essex, the Paris School of Economics, the Coalition Theory Network Workshop in Glasgow and the NOeg meeting in Vienna. Ortega acknowledges the hospitality of Columbia University. Any errors are our own. The corresponding code and data is freely available at \href{www.josueortega.com}{www.josueortega.com}. We have no conflict of interest nor external funding to disclose.}}

\date{\small First version: September 29, 2017. Revised: \today.}

\maketitle

\begin{abstract}
We used to marry people to whom we were somehow connected. Since we were more connected to people similar to us, we were also likely to marry someone from our own race. 

However, online dating has changed this pattern; people who meet online tend to be complete strangers. We investigate the effects of those previously absent ties on the diversity of modern societies. 

We find that social integration occurs rapidly when a society benefits from new connections. Our analysis of state-level data on interracial marriage and broadband adoption (proxy for online dating) suggests that this integration process is significant and ongoing.
\end{abstract}



\thispagestyle{empty}
\pagebreak

\section{Introduction}
In the most cited article on social networks,\footnote{\href{http://blogs.lse.ac.uk/impactofsocialsciences/2016/5/12/what-are-the-most-cited-publications-in-the-social-sciences-according-to-google-scholar/}{``What are the most-cited publications in the social sciences according to Google?''}, {\it LSE Blog}, 12/05/2016.} \citet{granovetter1973} argued that the most important connections we have may not be with our close friends but our acquaintances: people who are not very close to us, either physically or emotionally, help us to relate to groups that we otherwise would not be linked to. For example, it is from acquaintances that we are more likely to hear about job offers \citep{rees1966, corcoran1980, granovetter1995}. Those weak ties serve as bridges between our group of close friends and other clustered groups, hence allowing us to connect to the global community in a number of ways.\footnote{Although most people find a job via weak ties, it is also the case that weak ties are more numerous. However, the individual value from an additional strong tie is larger than the one from an additional weak tie \citep{kramarz2014,gee2017}.} 

Interestingly, the process of how we meet our romantic partners in at least the last hundred years closely resembles this phenomenon. We would probably not marry our best friends, but we are likely to end up marrying a friend of a friend or someone we coincided with in the past. \citet{rosenfeld2012} show how Americans met their partners in recent decades, listed by importance: through mutual friends, in bars, at work, in educational institutions, at church, through their families, or because they became neighbors. This is nothing but the weak ties phenomenon in action.\footnote{\citet{backstrom2014} reinforce the previous point: given the social network of a Facebook user who is in a romantic relationship, the node which has the highest chances of being his romantic partner is, perhaps surprisingly, not the one who has most friends in common with him.}\footnote{Similarly, most couples in Germany met through friends (32\%), at the workplace or school (21\%), and bars and other leisure venues (20\%) \citep{potarca2017}.}

But in the last two decades, the way we meet our romantic partners has changed dramatically. Rosenfeld and Thomas argue that {\it``the Internet increasingly allows Americans to meet and form relationships with perfect strangers, that is, people with whom they had no previous social tie''}. To this end, they document that in the last decade online dating\footnote{We use the term online dating to refer to any romantic relationship that starts online, including but not limited to dating apps. We use this terminology throughout the text.} has become the second most popular way to meet a spouse for Americans (see Figure \ref{fig:rosenfeld}).
\begin{figure}[!htbp]
	\centering
	\includegraphics[width=.97\textwidth]{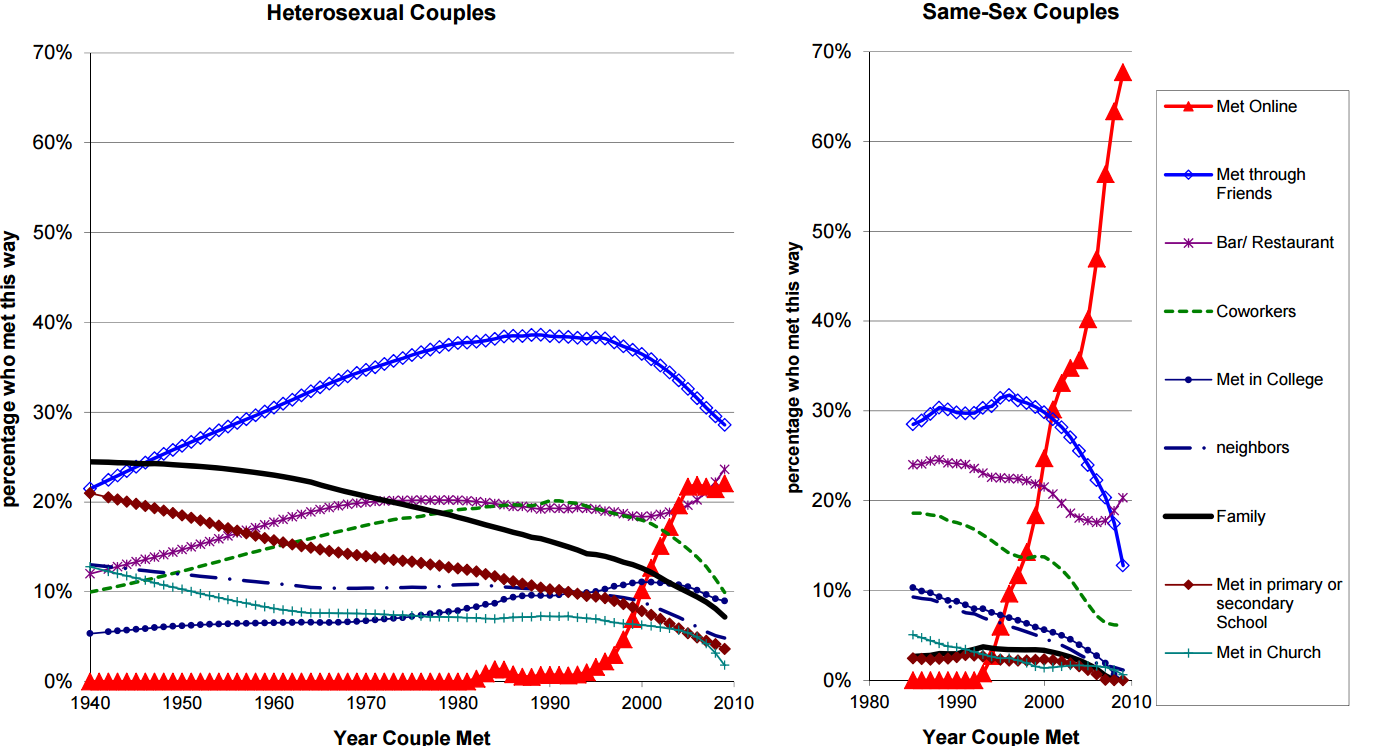}
	\caption{How we met our partners in previous decades.}
	\label{fig:rosenfeld}
	\begin{tablenotes}[Source] \centering
		\item \footnotesize Source: Rosenfeld and Thomas, 2012.
	\end{tablenotes}
\end{figure}		

Online dating has changed the way people meet their partners not only in America but in many places around the world. As an example, Figure \ref{fig:josuefriends} shows one of the author's Facebook friends graph. The yellow triangles reveal previous relationships that started in offline venues. It can clearly be seen that those ex-partners had several mutual friends with the author. In contrast, nodes appearing as red stars represent partners he met through online dating. These individuals have no contacts in common with him, and thus it is likely that, if it were not for online dating, they would have never interacted with him.\footnote{Although admittedly some of those links may have created after dating, what is clear is that the author was somewhat connected to these agents beforehand by some mutual connections, i.e. Granovetter's weak ties.}
\begin{figure}[!htbp]
	\centering
	\includegraphics[width=.97\textwidth]{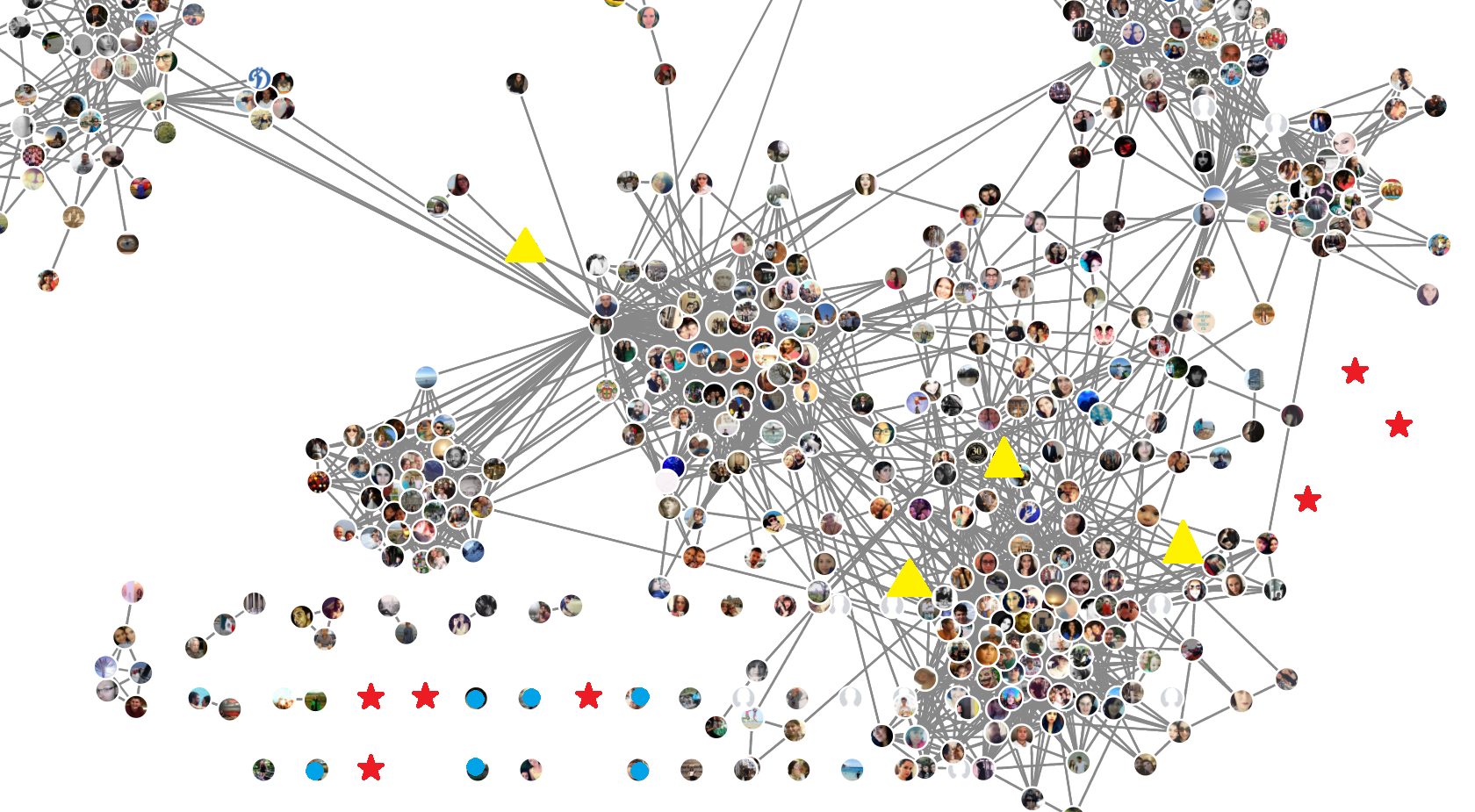}
	\caption{How one of us met his partners in the last decade.}
	\label{fig:josuefriends}
		\begin{tablenotes} \centering
			\item \footnotesize	Triangles are partners met offline, whereas starts are partners met online.
					\end{tablenotes}
\end{figure}	

Because a third of modern marriages start online \citep{cacioppo2013}, and up to 70\% of homosexual relationships, the way we match online with potential partners shapes the demography of our communities, in particular its racial diversity. Meeting people outside our social network online can intuitively increase the number of interracial marriages in our societies, which is remarkably low. Only 6.3\%  and 9\% of the total number of marriages are interracial in the US and the UK, respectively.\footnote{\href{http://www.pewresearch.org/fact-tank/2015/6/12/interracial-marriage-who-is-marrying-out/}{``Interracial marriage: Who is marrying out''}, {\it Pew Research Center}, 12/6/2015; and \href{http://webarchive.nationalarchives.gov.uk/20160107132103/http://www.ons.gov.uk/ons/dcp171776_369571.pdf}{``What does the 2011 census tell us about inter-ethnic relationships?''}, {\it UK Office for National Statistics}, 3/7/2014.} The low rates of interracial marriage are expected, given that up until 50 years ago these were illegal in many parts of the US, until the Supreme Court outlawed anti-miscegenation laws in the famous {\it Loving vs. Virginia} case.\footnote{Interracial marriage in the US has increased since 1970, but it remains rare \citep{arrow1998,kalmijn1998,fryer2007,furtado2015,chiappori2016black}. It occurs far less frequently than interfaith marriages \citep{qian1997}.}
	
This paper aims at improving our understanding of the impact of online dating on racial diversity in modern societies. In particular, we intend to find out how many more interracial marriages, if any, occur after online dating becomes available in a society. In addition, we are also interested in whether marriages created online are any different from those that existed before.

Understanding the evolution of interracial marriage is an important problem, for intermarriage is widely considered a measure of social distance in our societies \citep{wong2003,fryer2007,furtado2015}, just like residential or school segregation. Moreover, the number of interracial marriages in a society has important economic implications. It increases the social network of both spouses who intermarry by connecting them to people of different race. These valuable connections translate into a higher chance of finding employment \citep{furtado2010}.\footnote{There is a large literature that analyzes the effect of marrying an immigrant. This literature is relevant because often immigrants are from different races than natives. This literature has consistently found that an immigrant who married a native often has a higher probability of finding employment \citep{meng2005,furtado2010,goel2009}. Marrying a native increases the probability of employment, but not the perceived salary \citep{kantarevic2004}.} This partially explains why the combined income of an White-Asian modern couple is 14.4\% higher than than the combined income of an Asian-Asian couple, and 18.3\% higher than a White-White couple \citep{pew2012}. Even when controlling for factors that may influence the decision to intermarry, \cite{gius2013} finds that all interracial couples not involving African Americans have higher combined incomes than a White-White couple.\footnote{In some cases, intermarriage may even be correlated with poor economic outcomes. Examining the population in Hawaii, \cite{fu2007} finds that White people are 65\% more likely to live in poverty if they marry outside their own race.}  

Interracial marriage also affects the offspring of couples who engage in it. \cite{duncan2011} find that children of an interracial marriage between a Mexican Latino and an interracial partner enjoy significant human capital advantages over children born from endogamous Mexican marriages in the US.\footnote{Although Hispanic is not a race, Hispanics do not associate with ``standard'' races. In the 2010 US census, over 19 million Latinos identified themselves as being of ``some other race''. See \href{http://www.nytimes.com/2012/01/14/us/for-many-latinos-race-is-more-culture-than-color.html}{``For many Latinos, racial identity is more culture than color''}, {\it New York Times}, 13/1/2012.}
 Those human capital advantages include a 50\% reduction in the high school dropout rate for male children.\footnote{
\cite{pearce2012} examines the behavioral well-being of children in inter and intraracial households. They find no significant differences between the two groups.}

\subsection{Overview of Results}
This article builds a novel theoretical framework to study matching problems under network constraints. Our model is different to the previous theoretical literature on marriage in that we explicitly study the role of social networks in the decision of whom to marry. Consequently, our model provides new testable predictions regarding how changes in the structure of agents' social networks impact the number of interracial marriages and the quality of marriage itself. In particular, our model combines non-transferable utility matching  {\it \`a la} \citet{gale1962} with random graphs, first studied by \cite{gilbert1959} and \cite{erdos1959}, which we use to represent social networks.\footnote{Most of the literature studying marriage with matching models uses transferable utility, following the seminal work of \citet{becker1973,becker1981}. A review of that literature appears in \cite{browning2014}. Although our model departs substantially from this literature, we point out similarities with particular papers in Section \ref{sec:model}.}

We consider a society composed of agents who belong to different races. All agents want to marry the potential partner who is closest to them in terms of personality traits, but they can only marry people who they know, i.e. to whom they are connected. As in real life, agents are highly connected with agents of their own race, but only poorly so with people from other races.\footnote{The average American public school student has less than one school friend of another race \citep{fryer2007}. Among White Americans, 91\% of people comprising their social networks are also White, while for Black and Latino Americans the percentages are 83\% and 64\%, respectively \citep{values2}. \cite{patacchini2016} document that 84\% of the friends of white American students are also white. For high school students, less than 10\% of interracial friendships exist \citep{shrum1988}. Furthermore, only 8\% of Americans have anyone of another race with whom they discuss important matters \citep{marsden1987}.} Again inspired by empirical evidence \citep{hitsch2010,banerjee2013}, we assume that the marriages that occur in our society are those predicted by game-theoretic stability, i.e. no two unmarried persons can marry and make one better off without making the other worse off. In our model, there is a unique stable marriage in each society (Proposition \ref{prop:uniqueness}).

After computing the unique stable matching, we introduce online dating in our societies by creating previously absent ties between races and compute the stable marriage again.\footnote{We obtain the same qualitative results if we increase both interracial and intraracial connections, because the marginal value of interracial connections is much larger; see Appendix \ref{app:robust}. On a related note, although some dating websites allow the users to sort partners' suggestions based on ethnicity, many of them suggest partners randomly. For our main result, we only need that online daters meet at one partner outside their social circle. Rosenfeld and Thomas (2012) suggest that this is indeed the case.} We compare how many more interracial marriages are formed in the new expanded society that is more interracially connected. We also keep an eye on the characteristics of those newly formed marriages. In particular, we focus on the average distance in personality traits between partners before and after the introduction of online dating, which we use as a proxy for the strength of marriages in a society (ideally, all agents marry someone who has the same personality traits as them).

Perhaps surprisingly, we find that making a society more interracially connected may decrease the number of interracial marriages. It also may increase the average distance between spouses, and even lead to less married people in the society (Proposition \ref{prop:monotonicity}). However, this only occurs in rare cases. Our main result affirms that the expected number of interracial marriages in a society increases rapidly after new connections between races are added (Result \ref{res:diversity}). In particular, if we allow marriage between agents who have a friend in common, complete social integration occurs when the probability of being directly connected to another race is $\frac{1}{n}$, where $n$ is the number of persons in each race. This result provides us with our first and main testable hypothesis: social integration occurs rapidly after the emergence of online dating, even if the number of partners that individuals meet from newly formed ties is small. The increase in the number of interracial marriages in our model does not require changes in agents' preferences.

Furthermore, the average distance between married couples becomes smaller, leading to better marriages in which agents obtain more desirable partners on average (Result \ref{res:strength}). This second result provides another testable hypothesis: marriages created in a society with online dating last longer and report higher levels of satisfaction than those created offline. We find this hypothesis interesting, as it has been widely suggested that online dating creates relationships of a lower quality.\footnote{\href{https://nypost.com/2015/08/16/tinder-is-tearing-apart-society/}{``Tinder is tearing society apart''}, {\it New York Post}, 16/08/2015; and \href{https://www.theguardian.com/commentisfree/2011/jul/25/online-dating-love-product}{``Online dating is eroding humanity''}, {\it The Guardian}, 25/07/11.} Finally, the added connections in general increase the number of married couples whenever communities are not fully connected or are unbalanced in their gender ratio (Result \ref{res:size}). This result provides a third and final testable hypothesis: the emergence of online dating leads to more marriages.

We contrast the testable hypotheses generated by the model with US data. With regards to the first and main hypothesis, we find that the number of interracial marriages substantially increases after the popularization of online dating. This increase in interracial marriage cannot be explained by changes in the demographic composition of the US only, because Black Americans are the racial group whose rate of interracial marriage has increased the most, going from 5\% in 1980 to 18\% in 2015 \citep{pew2017}. However, the fraction of the US population that is Black has remained relatively constant during the last 50 years at around 12\% of the population \citep{pew2015}.

To properly identify the impact of online dating on the generation of new interracial marriages, we exploit sharp temporal
and geographic variation in the pattern of broadband adoption, which we use as a proxy for the introduction of online dating. This strategy is sensible given that broadband adoption has limited correlation to other factors impacting interracial marriages and eliminates the possibility of reverse causation. Using this data from 2000 to 2016, we conclude that one additional line of broadband internet 3 years ago (marriages take time) affects the probability of being in an interracial marriage by 0.07\%. We obtain this effect by estimating a linear probability model that includes a rich set of individual- and state-level controls, including the racial diversity of each state and many others. Therefore, we conclude that there is evidence suggesting that online dating is causing more interracial marriages, and that this change is ongoing.

Moreover, shortly after we first made available our paper online on September of 2017, \cite{thomas2018} used recently collected  data on how couples meet to successfully demonstrate that couples that met online are more likely to be interracial, even when controlling for the diversity of their corresponding locations. Thomas estimates that American couples who met online since 1996 are 6\% to 7\% more likely to be interracial than those who met offline. His findings further establish that online dating has indeed had a positive impact on the number of interracial marriages, as predicted by our model.

With respect to the quality of marriages created online, both \cite{cacioppo2013} and \cite{rosenfeld2017} find that relationships created online last at least as long as those created offline, defying the popular belief that marriages that start online are of lower quality than those that start elsewhere, and are in line with our second prediction (in fact, \citealp{cacioppo2013} finds that marriages that start online last longer and report a higher marital satisfaction).\footnote{Because online dating is a recent phenomenon, it is unclear whether these effects will persist in the long run. However, the fact that independent studies find similar effects suggests that these findings are robust. \cite{rosenfeld2017} also finds that couples who meet online marry faster than those created offline. }

Finally, with respect to our third hypothesis that asserts that online dating should increase the number of married couples, \citet{bellou2015} finds causal evidence that online dating has increased the rate at which both White and Black young adults marry in the US. The data she analyzes shows that online dating has contributed to higher marriage rates by up to 33\% compared to the counterfactual without internet dating. Therefore, our third prediction is consistent with Bellou's findings.

\subsection{Structure of the Article}
We present our model in Section \ref{sec:model}. Section \ref{sec:measures} introduces the welfare indicators underlying the further analysis. Sections \ref{sec:theory} and \ref{sec:simulations} analyze how these measures change when societies become more connected using theoretical analysis and simulations, respectively. Section \ref{sec:data} contrasts our model predictions with observed demographic trends from the US. Section \ref{sec:conclusion} concludes.

\section{Model}
\label{sec:model}
\subsection{Agents}
There are $r$ races or communities, each with $n$ agents (also called nodes). Each agent $i$ is identified by a pair of coordinates $(x_i,y_i) \in [0,1]^2$, that can be understood as their {\it personality traits} (e.g. education, political views, weight, height, etcetera).\footnote{For a real-life representation using a 2-dimensional plane see {\tt www.politicalcompass.org}. A similar interpretation appears in \citet{chiappori2012} and in \citet{chiappori2016}. We use two personality traits because it allows us to use an illustrative and pedagogic graphical representation. All the results are robust to adding more personality traits.} Both coordinates are drawn uniformly and independently for all agents. Each agent is either male or female. Female agents are plotted as stars and males as dots. Each race has an equal number of males and females, and is assigned a particular color in our graphical representations.

\subsection{Edges}
Between any two agents of the same race, there exists a connecting edge (also called link) with probability $p$: these edges are represented as solid lines and occur independently of each other. Agents are connected to others of different race with probability $q$: these interracial edges appear as dotted lines and are also independent. The intuition of our model is that two agents know each other if they are connected by an edge.\footnote{This interpretation is common in the study of friendship networks, see \cite{demarti2016} and references therein. Our model can be understood as the islands model in \cite{golub2012}, in which agents' type is both their race and gender.} We assume that $p>q$.  We present an illustrative example in Figure \ref{fig:dots}. 

\begin{figure}[!htbp]
\centering
	\includegraphics[width=1\textwidth]{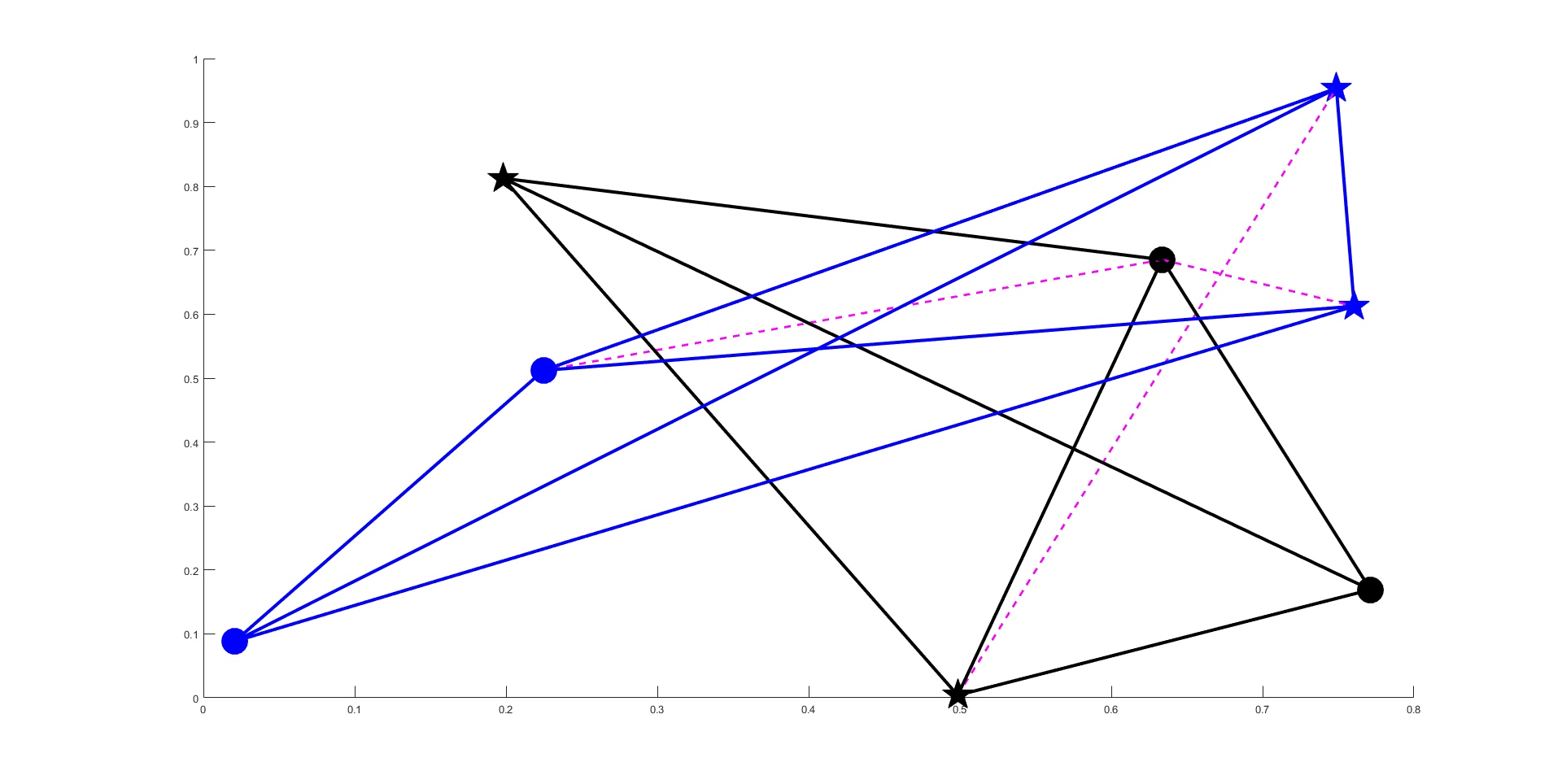}
	\caption{Example of a society with $n=4$ agents, $r=2$ races, $p=1$ and $q=0.2$.}
	\label{fig:dots}
\end{figure}	

Our model is a generalization of the random graph model (\citealp{erdos1959, gilbert1959}; for a textbook reference, see \citealp{bollobas2001}). Each race is represented by a random graph with $n$ nodes connected among them with probability $p$. Nodes are connected across graphs with probability $q$. The $r$ random graphs are the within-race set of links for each race. In expectation, each agent is connected to $n(r-1)q+(n-1)p$ persons.

A {\it society} $S$ is a realization from a generalized random graph model, defined by a four-tuple ($n,r,p,q$). A society $S$ has a corresponding graph $S=(M \cup W;E)$, where $M$ and $W$ are the set of men and women, respectively, and $E$ is the set of edges. We use the notation $E(i,j)=1$ if there is an edge between agents $i$ and $j$, and 0 otherwise. We denote such edge by either $(i,j)$ or $(j,i)$.  

\subsection{Agents' Preferences}
All agents are heterosexual and prefer marrying anyone over remaining alone.\footnote{Both assumptions are innocuous and for exposition only.} We denote by $P_i$ the set of potential partners for $i$, i.e. those of different gender. The preferences of agent $i$ are given by a function $\delta_i:P_i \to \RR_+$ that has a distance interpretation.\footnote{The function $\delta$ can be generalized to include functions that violate the symmetry $(\delta(x,y)\neq \delta(y,x))$ and identity ($\delta(x,x)=0$) properties of mathematical distances.} An agent $i$ prefers agent $j \in P_i$ over agent $k \in P_i$ if $\delta_i(i,j) \leq \delta_i(i,k)$. The intuition is that agents like potential partners that are close to them in terms of personality traits. The function $\delta_i$ could take many forms, however we put emphasis on two intuitive ones.

The first one is the Euclidean distance for all agents, so that for any agent $i$ and every potential partner $j \neq i$, 
\begin{equation}
\delta^E(i,j)= \sqrt{(x_i - x_j)^2 + (y_i - y_j)^2}
\end{equation}

and $\delta^E(i,i)=\sqrt 2$ $\forall i \in M \cup W$, i.e. the utility of remaining alone equals the utility derived from marrying the worst possible partner. Euclidean preferences are intuitive and have been widely used in the social sciences \citep{bogomolnaia2007}. The indifference curves associated with Euclidean preferences can be described by concentric circles around each node. 

The second preferences we consider are such that every agent prefers a partner close to them in personality trait $x$, but they all agree on the optimum value in personality trait $y$. The intuition is that the $y$-coordinate indicates an attribute that is usually considered desirable by all partners, such as wealth. We call these preferences {\it assortative}.\footnote{If we keep the $x$-axis fixed, so that agents only care about the $y$-axis, we get full assortative mating as a particular case. The preferences for the $y$ attribute are also known as vertical preferences.} Formally, for any agent $i$ and every potential partner $j \in P_i$, 
\begin{equation}
\delta^A(i,j)= \left\vert x_i - x_j \right\vert+(1-y_j)
\end{equation}

and $\delta^E(i,i)=2$ $\forall i \in M \cup W$. The $\delta$ functions we discussed can be weighted to account for the strong intraracial preferences that are often observed in reality \citep{wong2003,fisman2008,hitsch2010,rudder2014,potarca2015,mcgrath2016}.\footnote{It is not clear whether the declared intraracial preferences show an intrinsic intraracial predilection or capture external biases, which, when removed, leave the partner indifferent to match across races. Evidence supporting the latter hypothesis includes: \citet{fryer2007} documents that White and Black US veterans have had higher rates of intermarriage after serving with mixed communities. \citet{fisman2008} finds that people do not find partners of their own race more attractive. \citet{rudder2009} shows that online daters have a roughly equal user compatibility. \citet{lewis2013} finds that users are more willing to engage in interracial dating if they previously interacted with a dater from another race.} Inter or intraracial preferences can be incorporated into the model, as in equation (\ref{eq:weights}) below
\begin{equation}
\delta'_i(i,j)=\beta_{ij} \; \delta(i,j)
\label{eq:weights}
\end{equation}

where $\beta_{ij}=\beta_{ik}$ if agents $j$ and $k$ have the same race, and $\beta_{ij} \neq \beta_{ik}$ otherwise. In equation (\ref{eq:weights}), the factor $\beta_{ij}$ captures weightings in preferences. The case $\beta_{ij}<1$ implies that agent $i$ relative prefers potential partners of the same race as agent $j$, while $\beta_{ij}>1$ expresses relative dislike towards potential partners of the same race as agent $j$. Although our results are qualitatively the same when we explicitly incorporate racial preferences using the functional form in equation (\ref{eq:weights}), we postpone this analysis to Appendix \ref{app:robust}. 

A society in which all agents have either all Euclidean or all assortative preferences will be called Euclidean or assortative, respectively. We focus on these two cases. In both cases agents' preferences are strict because we assume personality traits are drawn from a continuous distribution.

\subsection{Marriages}
Agents can only marry potential partners who they know, i.e. if there exists a path of length at most $k$ between them in the society graph.\footnote{A path from node $i$ to $t$ is a set of edges $(ij),(jk),\ldots,(st)$. The length of the path is the number of such pairs.} We consider two types of marriages:
\begin{enumerate}[topsep=0pt,itemsep=-1ex,partopsep=1ex,parsep=1ex]
	\item Direct marriages: $k=1$. Agents can only marry if they know each other. 
	\item Long marriages: $k=2$. Agents can only marry if they know each other or if they have a mutual friend in common.
\end{enumerate}

To formalize the previous marriage notion, let $\rho_k(i,j)=1$ if there is a path of at most length $k$ between $i$ and $j$, with the convention $\rho_1(i,i)=1$. A marriage $\mu:M \cup W \to M \cup W$ of length $k$ is a function that satisfies
\begin{eqnarray}
\forall m\in M && \mu(m) \in W \cup \{m\} \\
\forall w\in W && \mu(w) \in M \cup \{w\} \\
\forall i \in M \cup W && \mu( \mu(i) )=i\\
\forall i \in M \cup W && \mu( i )=j \text{ only if } \rho_k(i,j)=1
\end{eqnarray}

We use the convention that agents that remain unmarried are matched to themselves. Because realized romantic pairings are close to those predicted by stability \citep{hitsch2010, banerjee2013}, we assume that marriages that occur in each society are {\it stable}.\footnote{We study the stability of the marriages created, following the matching literature, not of the network per se. Stability of networks was defined by \cite{jackson1996} in the context of network formation. We take the network structure as given.} A marriage $\mu$ is $k$-stable if there is no man-woman pair $(m,w)$ who are not married to each other such that
\begin{eqnarray}
\label{eq:block} \rho_k(m,w)&=&1\\
\delta(m,w) &<& \delta(m,\mu(m))\\
\delta(w,m) &<& \delta(w,\mu(w))
\end{eqnarray} 

Such a pair is called a blocking pair. Condition (\ref{eq:block}) is the only non-standard one in the matching literature, and ensures that a pair of agents cannot block a direct marriage if they are not connected by a path of length at most $k$ in the corresponding graph. Given our assumptions regarding agents' preferences,

\begin{proposition}
\label{prop:uniqueness}
For any positive integer $k$, every Euclidean or assortative society has a unique $k$-stable marriage.
\end{proposition}

\begin{proof}For the Euclidean society, a simple algorithm computes the unique $k$-stable marriage. Let every person point to their preferred partner to whom they are connected to by a path of length at most $k$. In case two people point to each other, marry them and remove them from the graph. Let everybody point to their new preferred partner to which they are connected to among those still left. Again, marry those that choose each other, and repeat the procedure until no mutual pointing occurs. The procedure ends after at most $\frac{rn}{2}$ iterations. This algorithm is similar to the one proposed by \citet{holroyd2009} for 1-stable matchings in the mathematics literature\footnote{\cite{holroyd2009} require two additional properties: non-equidistance and no descending chains. The first one is equivalent to strict preferences, the second one is trivially satisfied. In their algorithm, agents point to the closest agent, independently if they are connected to them.} and to David Gale's top trading cycles algorithm (in which agents' endowments are themselves), used in one-sided matching with endowments \citep{shapley1974} .

For the assortative society, assume by contradiction that there are two $k$-stable matchings $\mu$ and $\mu'$ such that for two men $m_1$ and $m_2$, and two women $w_1$ and $w_2$, $\mu(m_1)=w_1$ and $\mu(m_2)=w_2$, but $\mu'(m_1)=w_2$ and $\mu'(m_2)=w_1$.\footnote{It could be the case that in the two matchings there are no four people who change partner, but that the swap involves more agents. The argument readily generalizes.} The fact that both marriages are $k$-stable implies, without loss of generality, that for $i,j \in \{1,2\}$ and $i \neq j$, $\delta(m_i,w_i)-\delta(m_i,w_j)<0$ and $\delta(w_i,m_j)-\delta(w_i,m_i)<0$. Adding up those four inequalities, one obtains $0<0$, a contradiction. $\blacksquare$
\end{proof}
 
The fact that the stable marriage is unique allows us to unambiguously compare the characteristics of marriages in two different societies.\footnote{In general, the set of stable marriages is large. Under different restrictions on agents' preferences we also obtain uniqueness \citep{eeckhout2000, clark2006}. None of the restrictions mentioned in those papers applies the current setting.} Figure \ref{fig:marriages} illustrates the direct and long stable marriages for the Euclidean and assortative societies depicted in Figure \ref{fig:dots}. Marriages are represented as red thick edges.  

\begin{figure*}[t]
    \centering
        \begin{subfigure}[b]{0.485\textwidth}
            \centering
            \includegraphics[width=\textwidth]{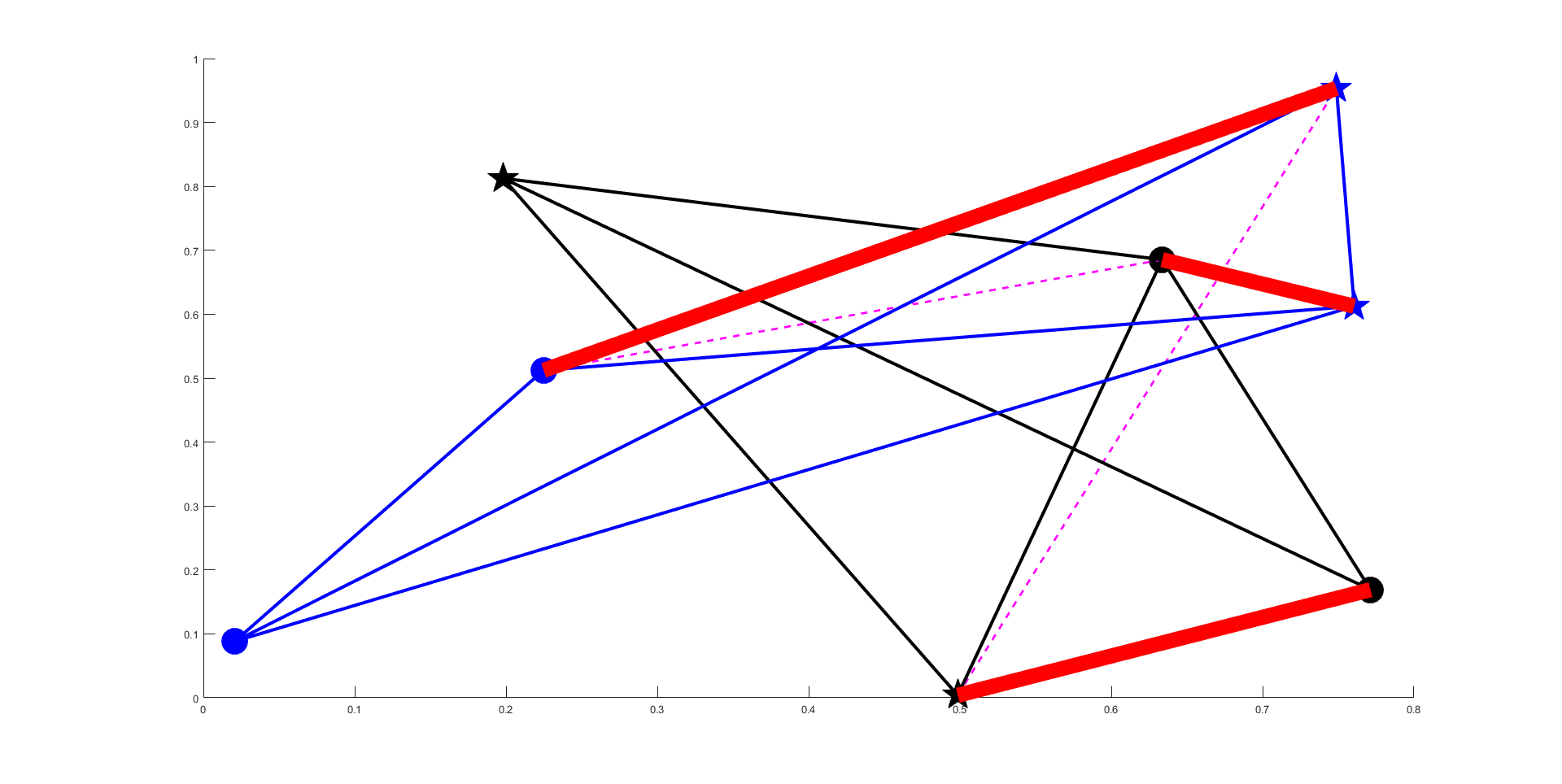}
             \caption{Direct marriage, Euclidean preferences.}
							\label{fig:w}
        \end{subfigure}
        \hfill
        \begin{subfigure}[b]{0.485\textwidth}  
            \centering 
            \includegraphics[width=\textwidth]{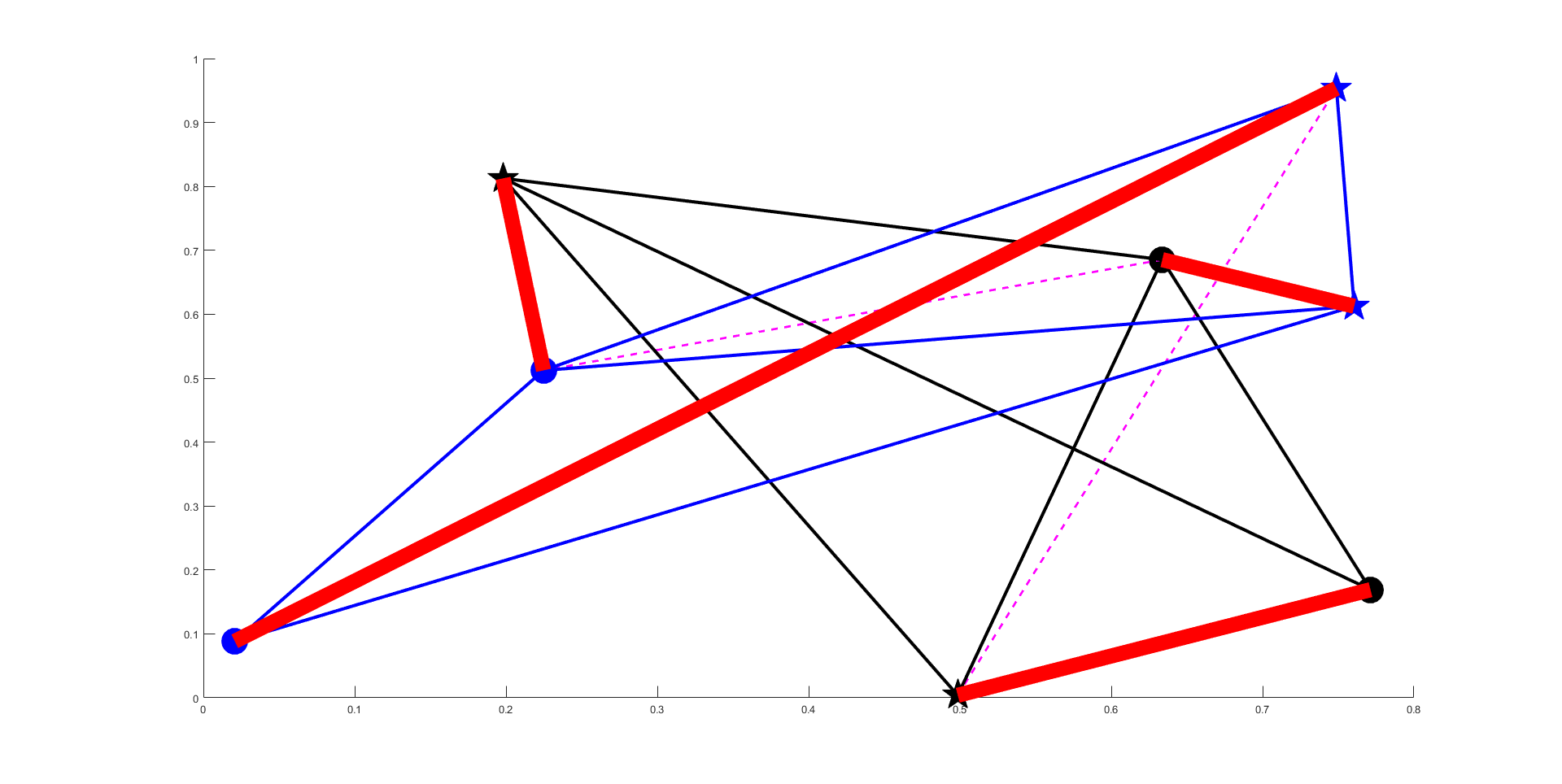}
             \caption{Long marriage, Euclidean preferences.}
							\label{fig:w1}
        \end{subfigure}
        \vskip\baselineskip
        \begin{subfigure}[b]{0.485\textwidth}   
            \centering 
            \includegraphics[width=\textwidth]{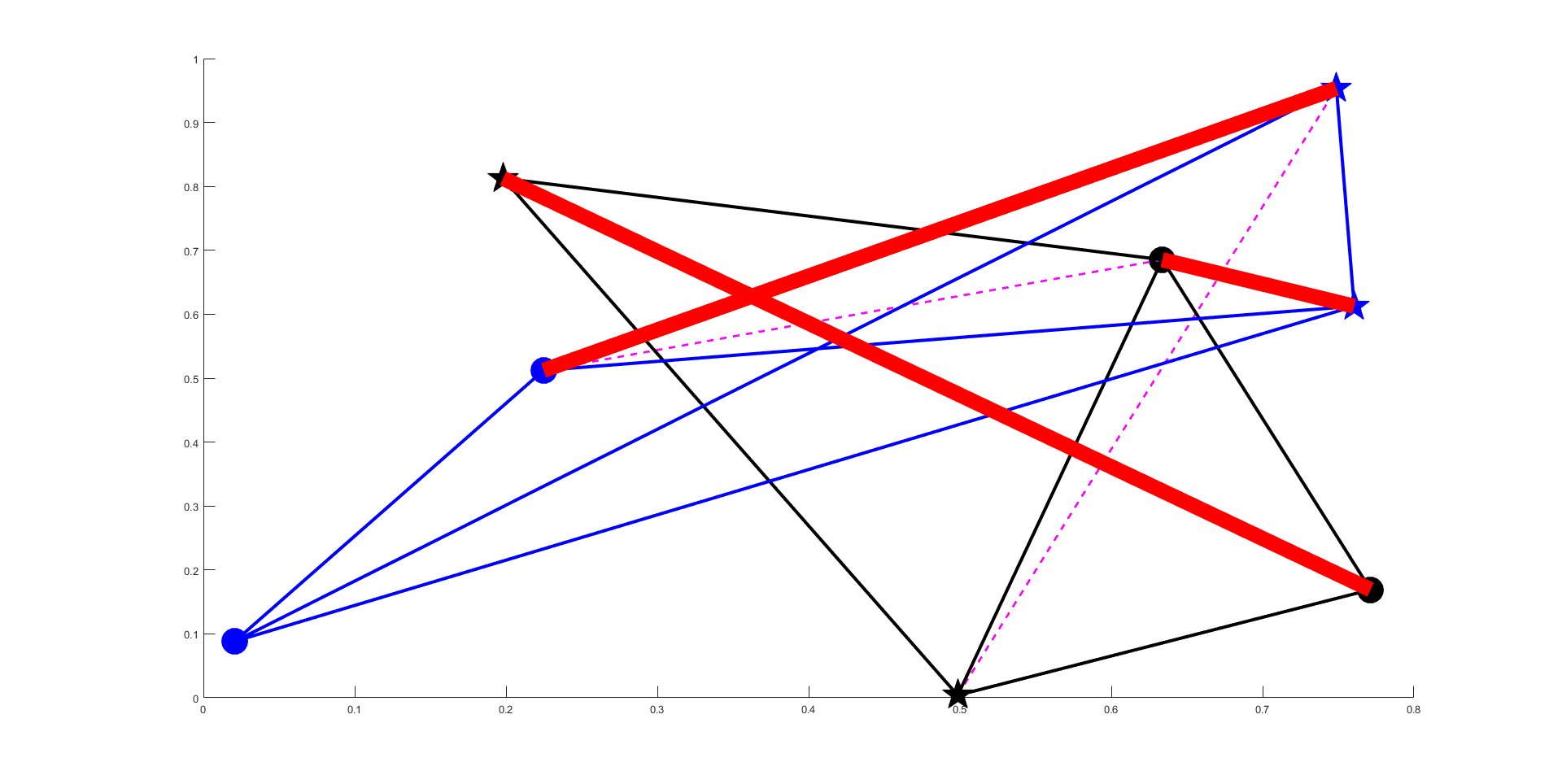}
             \caption{Direct marriage, assortative preferences.}
        \label{fig:s}
        \end{subfigure}
        \hfill
        \begin{subfigure}[b]{0.485\textwidth}   
            \centering 
            \includegraphics[width=\textwidth]{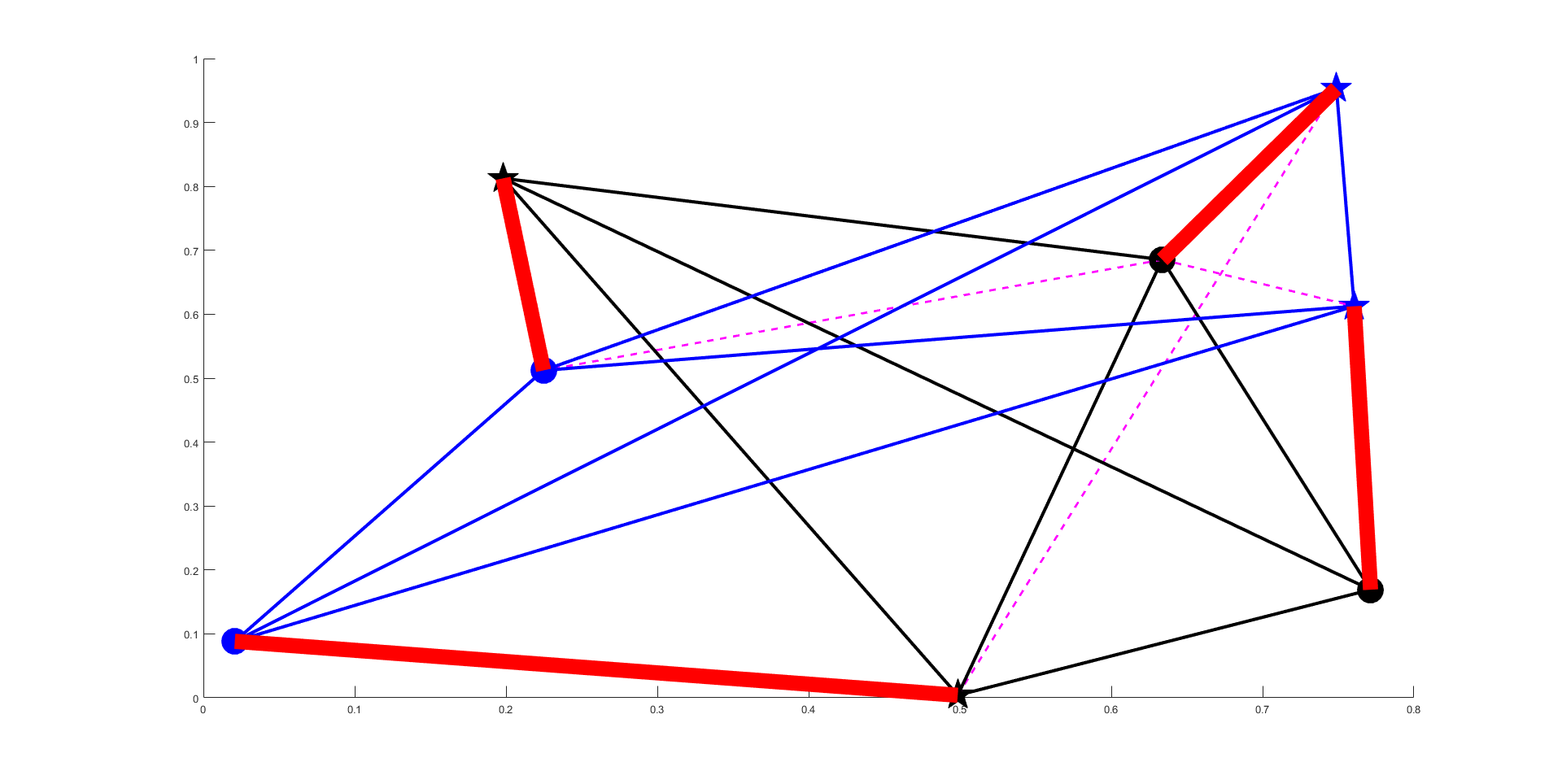}
             \caption{Long marriage, assortative preferences.}
        \label{fig:s2}
        \end{subfigure}
    \caption{Direct and long stable marriages for the Example in Figure \ref{fig:dots}.}
		\label{fig:marriages}
\end{figure*}

\subsection{Online Dating on Networks and Expansions of Societies}
We model online dating in a society $S$ by increasing the number of interracial edges. Given the graph $S=(M \cup W;E)$, we create new interracial edges between every pair that is disconnected with a probability $\epsilon$.\footnote{Online dating is likely to also increase the number of edges inside each race, but since we assume that $p>q$, these new edges play almost no role. We perform robustness checks in Appendix \ref{app:robust}, increasing both $p$ and $q$ but keeping its ratio fixed.}$\textcolor{blue}{^,}$\footnote{We could assume that particular persons are more likely than others to use online dating, e.g. younger people. However, the percentage of people who use online dating has increased for people of all ages. See: \href{http://www.pewresearch.org/fact-tank/2016/02/29/5-facts-about-online-dating/}{``5 facts about online dating''}, {\it Pew Research Center}, 29/2/2016. To obtain our main result, we only need a small increase in the probability of interconnection for each agent.} $S_\epsilon$ denotes a society that results after online dating has occurred in society $S$. $S_\epsilon$ has exactly the same nodes as $S$, and all its edges, but potentially more. We say that the society $S_\epsilon$ is an {\it expansion} of the society $S$. Equivalently, we model online dating by increasing $q$. Online dating generates a society drawn from a generalized random graph model with a higher $q$, i.e. with parameters ($n,r,p,q+ \epsilon$).

\section{Welfare Indicators}
\label{sec:measures}
We want to understand how the welfare of a society changes after online dating becomes available, i.e. after a society becomes more interracially connected. We consider three welfare indicators:

\noindent
\begin{enumerate}[leftmargin=0cm,itemindent=.5cm,labelwidth=\itemindent,labelsep=0cm,align=left]
	\item {\bf Diversity}, i.e. how many marriages are interracial. We normalize this indicator so that 0 indicates a society with no interracial marriages, and 1 equals the diversity of a colorblind society in which $p=q$, where an expected fraction $\frac{r-1}{r}$ of the marriages are interracial. Formally, let $\mathcal R$ be a function that maps each agent to their race and $M^*$ be the set of married men. Then
\begin{equation}
	dv(S)= \frac{\left\vert \{ m \in M^* : \mathcal R(m) \neq \mathcal R (\mu(m)) \} \right\vert }{  \left\vert M^* \right\vert } \cdot \frac{r}{r-1}
\end{equation}

	\item {\bf Strength}, defined as $\sqrt{2}$ minus the average Euclidean distance between each married couple. This number is normalized to be between 0 and 1. If every agent gets her perfect match, strength is 1, but if every agent marries the worst possible partner, strength equals 0. We believe strength is a good measure of the quality of marriage not only because it measures how much agents like their spouses, but also because a marriage with a small distance between spouses is less susceptible to break up when random agents appear. Formally
\begin{equation}
	st(S)=\frac{\sqrt{2} -  \frac{\sum_{m \in M^*} \delta^E  ( m,\mu(m))}{\left\vert M^* \right\vert} }{\sqrt{2}}
\end{equation}

	\item {\bf Size}, i.e. the ratio of the society that is married. Formally,
	\begin{equation}
	sz(S)=\frac{\left\vert M^* \right\vert}{n}
	\label{eq:size}
	\end{equation}
\end{enumerate}

\section{Edge Monotonicity of Welfare Indicators}
\label{sec:theory}
Given a society $S$, the first question we ask is whether the welfare indicators of a society always increase when its number of interracial edges grow, i.e. when online dating becomes available. We refer to this property as {\it edge monotonicity}.\footnote{Properties that are edge monotonic have been thoroughly studied in the graph theory literature \citep{erdos1995}. Edge monotonicity is different from node monotonicity, in which one node, with all its corresponding edges, is added to the matching problem. It is well-known that when a new man joins a stable matching problem, every woman weakly improves, while every man becomes weakly worse off (Theorems 2.25 and 2.26 in \citealp{roth1992}).} 

\begin{definition}A welfare measure $w$ is edge monotonic if, for any society $S$, and any of its extensions $S_\epsilon$, we have
\begin{equation}
w(S_\epsilon) \geq w(S) 
\end{equation}
\end{definition}

That a welfare measure is edge monotonic implies that a society unambiguously becomes better off after becoming more interracially connected. Unfortunately,
\begin{proposition}
\label{prop:monotonicity}
	Diversity, strength, and size are all not edge monotonic.
\end{proposition}

\begin{proof}
We show that diversity, strength and size are not edge monotonic by providing counterexamples. To show that size is not edge monotonic, consider the society in Figure \ref{fig:dots} and its direct stable matching in Figure \ref{fig:w}. Remove all interracial edges: it is immediate that in the unique stable matching there are now four couples, one more than when interracial edges are present. 

For the case of strength, consider a simple society in which all nodes share the same $y$-coordinate, as the one depicted in Figure \ref{fig:strength}. There are two intraracial marriages and the average Euclidean distance is 0.35. When we add the interracial edge between the two central nodes, the closest nodes marry and the two far away nodes marry too. The average Euclidean distance in the expanded society increases to 0.45, hence reducing its strength.
\begin{figure}[ht]
				\centering
        \includegraphics[width=.97\textwidth]{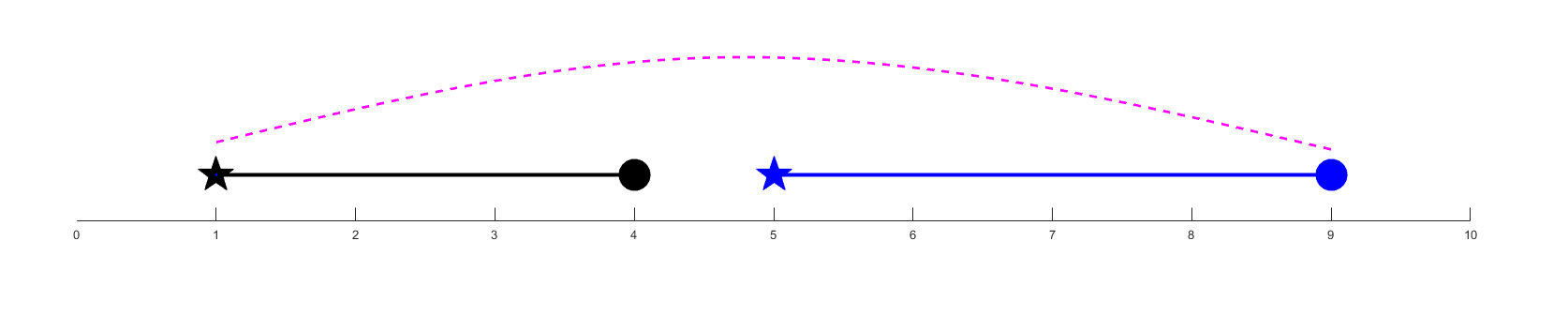}
				\caption{Strength is not edge monotonic.}
									\begin{tablenotes} \centering
		\item \footnotesize The average Euclidean distance between spouses increases after creating the interracial edge between the nodes in the center.
					\end{tablenotes}
				\label{fig:strength}
\end{figure}

To show that diversity is not edge monotonic, consider Figure \ref{fig:diversityem}. There are two men and two women of each of two races $a$ and $b$. Each gender is represented with the superscript $^+$ or $^-$.
\begin{figure}[ht]
    \centering 
				\includegraphics[width=\textwidth]{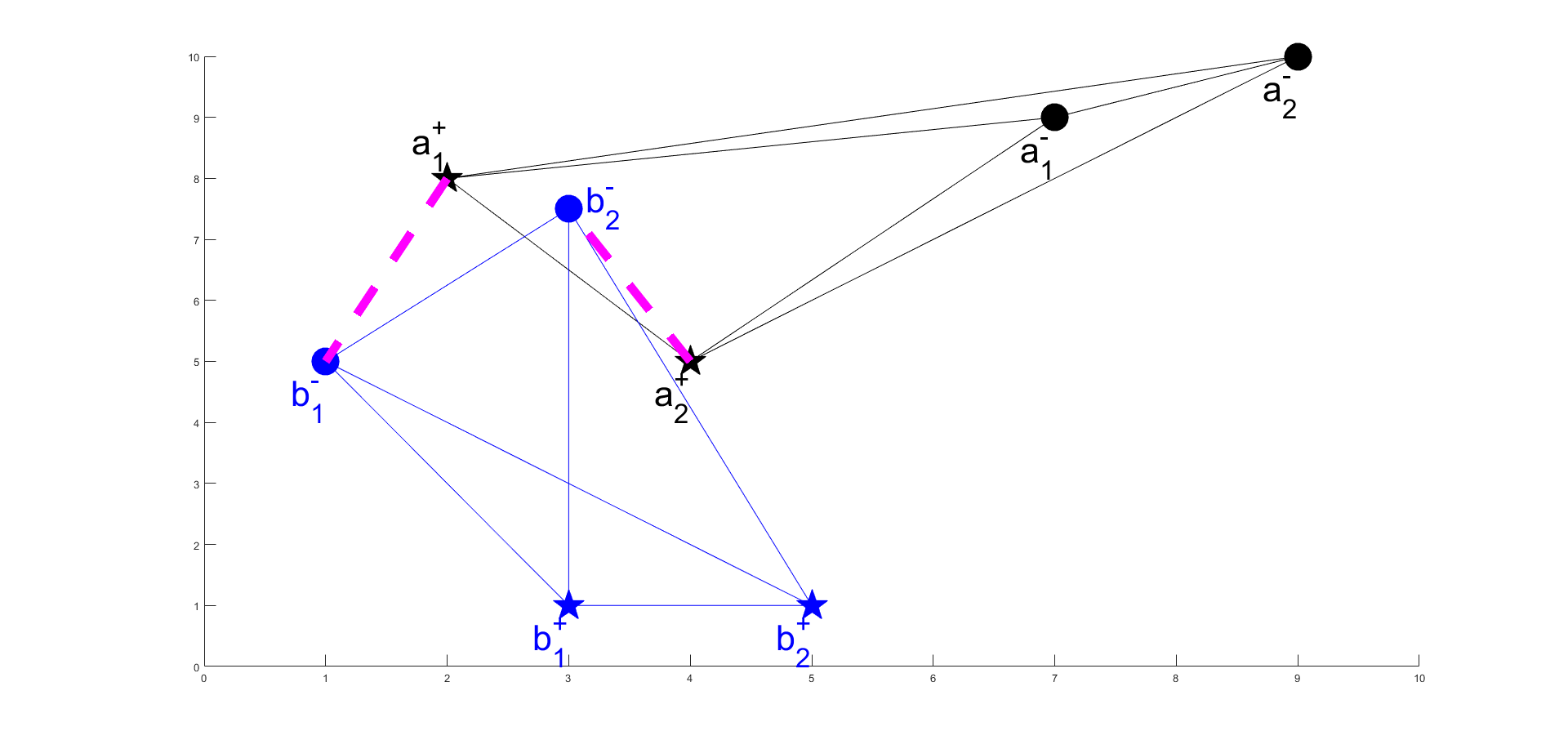}
    \begin{subfigure}[b]{0.485\textwidth}
        \includegraphics[width=\textwidth]{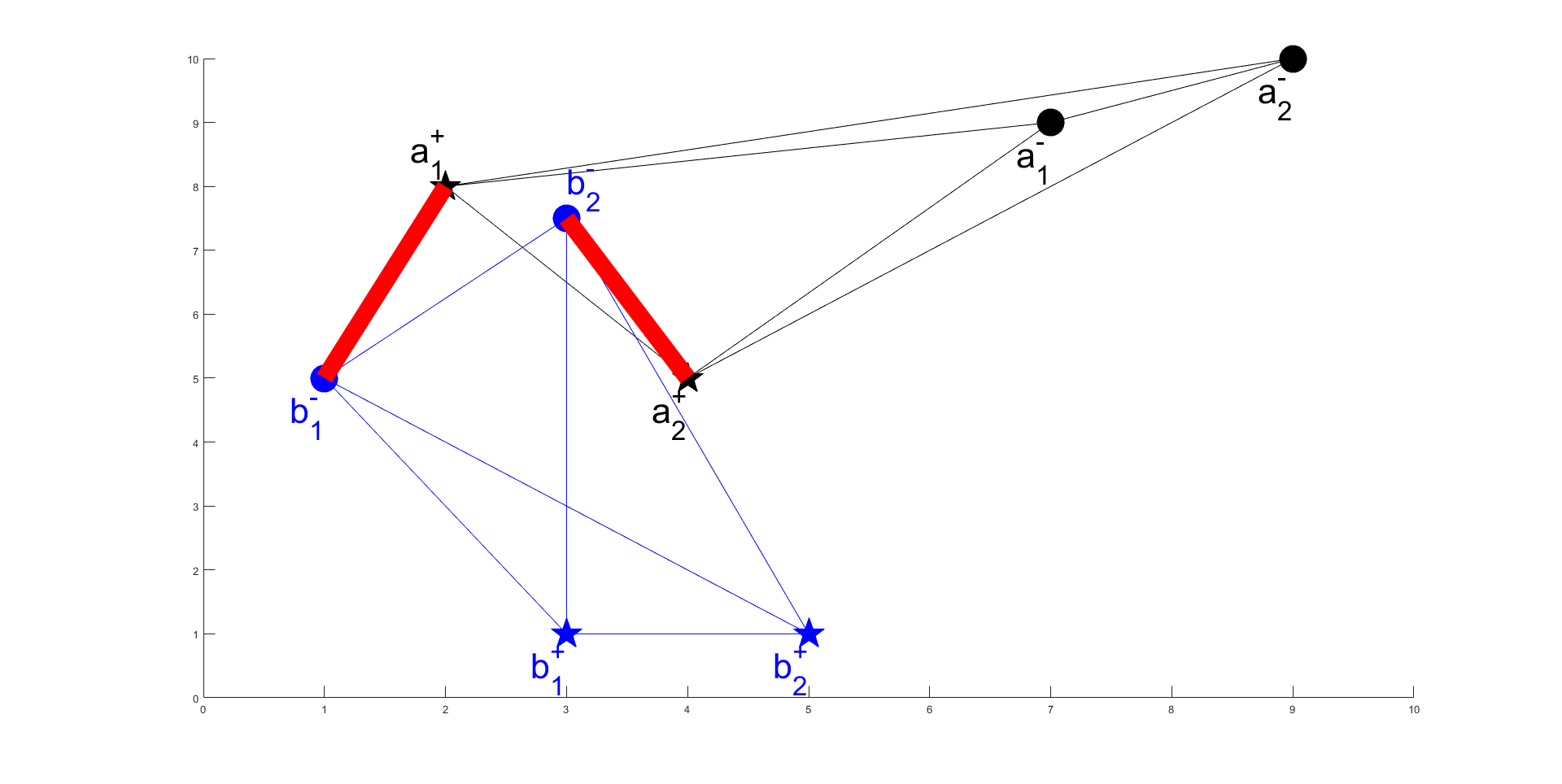}
        \caption{$dv(S)=2$}
        \label{fig:dem1}
    \end{subfigure}
    ~ 
    \begin{subfigure}[b]{0.485\textwidth}
        \includegraphics[width=\textwidth]{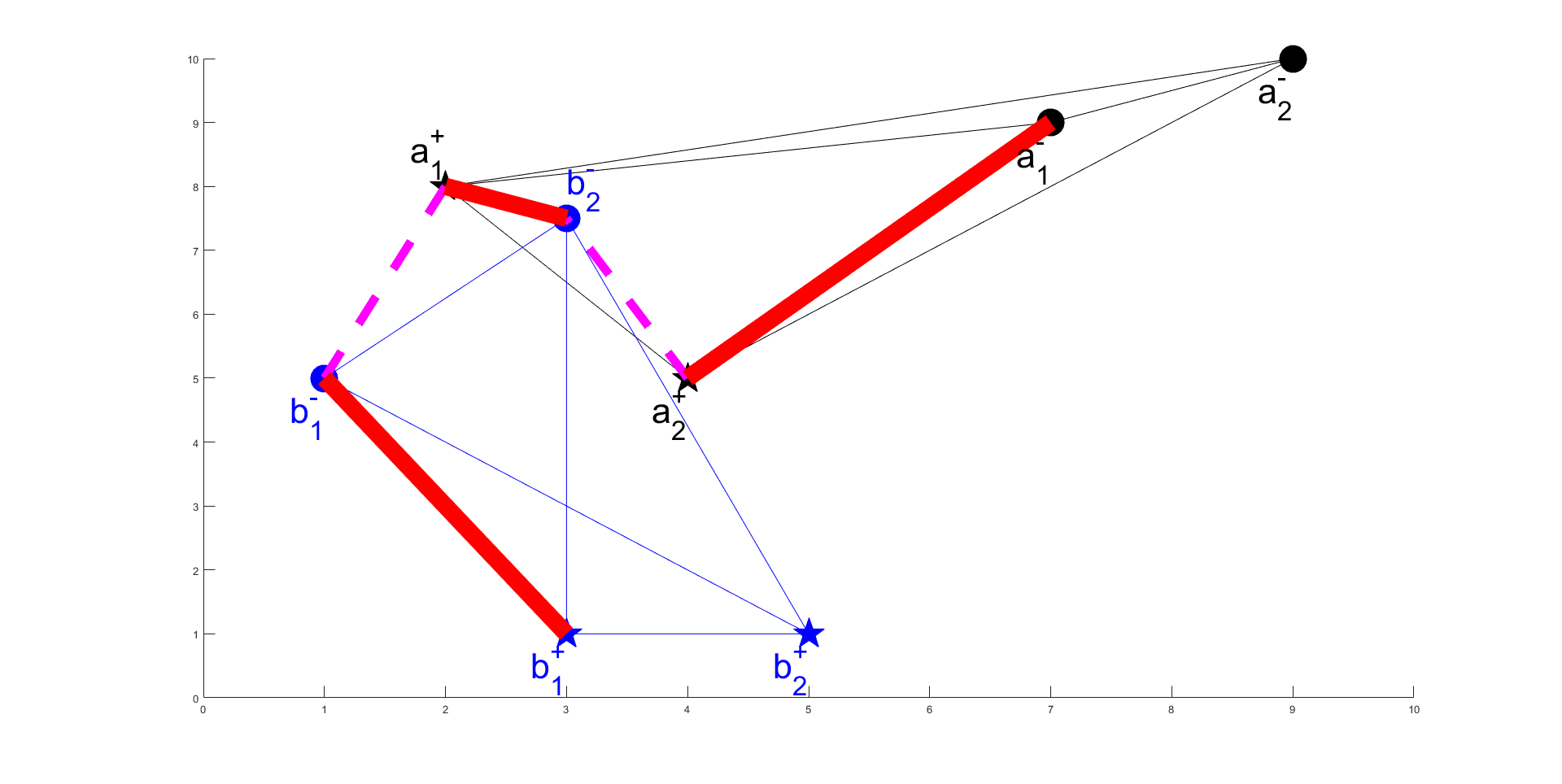}
        \caption{$dv(S_\epsilon)=2/3$}
        \label{fig:dem2}
    \end{subfigure}
				\caption{Diversity is not edge monotonic.}
					\begin{tablenotes} \centering
		\item \footnotesize The diversity of this society reduces after creating the interracial edge $(a_1^+,b_2^-)$. The top graph represents the original society. The bottom left graph show the marriages in the original society, whereas the bottom right graph shows the marriages in the expanded society.
					\end{tablenotes}
				\label{fig:diversityem}
\end{figure}

Stability requires that $\mu(b_1^-)=a_1^+$ and $\mu(b_2^+)=a_2^-$, and everyone else is unmarried. However, when we add the interracial edge $(a_1^+b_2^-)$, the married couples become $\mu(b_1^-)=b_1^+$, $\mu(a_2^+)=a_1^-$, and $\mu(a_1^+)=b_2^-$. In this extended society, there is just one interracial marriage, out of a total of three, when before we had two out of two. Therefore diversity reduces after adding the edge $(a_1^+b_2^-)$. $\blacksquare$
\end{proof}

The failure of edge monotonicity by our three welfare indicators makes evident that, to evaluate welfare changes in societies, we need to understand how welfare varies in an average society after introducing new interracial edges. We develop this comparison in the next Section.

A further comment on edge monotonicity. The fact that the size of a society is not edge monotonic implies that adding interracial edges may not lead to a Pareto improvement for the society. Some agents can become worse off after the society becomes more connected. Nevertheless, the fraction of agents that becomes worse off after adding an extra edge is never more than one-half of the society, and although it does not vanish as the societies grow large, the welfare losses measured in difference in spouse ranking become asymptotically zero. \cite{ortega2017} discusses both findings in detail.

\section{Expected Welfare Indicators}
\label{sec:simulations}

To understand how the welfare indicators behave on average, we need to form expectations of these welfare measures. We are able to evaluate this expression analytically for diversity, and rely on simulation results for the others. 

\subsection{Diversity}
The expected diversity of a society with direct marriages is given by
\begin{equation}
\EE[dv(S_{\text{direct}})]=\frac{q\frac{(r-1)n}{2}}{p\frac{n}{2}+q\frac{(r-1)n}{2}} \cdot \frac{r}{r-1}
\label{eq:expecteddiversity}
\end{equation}

where $q(r-1)n/2$ is the expected number of potential partners of a different race to which an agents is directly connected, and $pn/2$ is the corresponding expected number of potential partners of the same race. The term $\frac{r}{r-1}$ is just the normalization we impose to ensure that diversity equals one when $p=q$. Equation (\ref{eq:expecteddiversity}) is a concave function of $q$, because 
\begin{equation}
\frac{\partial^2 \EE[dv(S_{\text{direct}})]}{\partial q^2}=\frac{-pr(r-1)}{(pn+q(r-1)n)^3} < 0
\end{equation}
and therefore a small increase in $q$ around $q=0$ produces an even larger increment in the expected diversity of a society. If we consider long marriages, we observe a more interesting change. The expected diversity in a society with long marriages is given by
\begin{equation}
\EE[dv(S_{\text{long}})]= \frac{ P(B) \frac{(r-1)n}{2} }{ P(A) \frac{n}{2} +  P(B) \frac{(r-1)n}{2} } \cdot \frac{r}{r-1}
\label{eq:expecteddiversitylong}
\end{equation}
where $P(A)$ denotes the probability that any agent (say $i$) is connected to another member of his community ($i'$) by a path of length at most 2, and $P(B)$ denotes the probability that any agent ($i$) is connected to any agent of another community ($j$) by a path of length at most two, perhaps via another agent ($h$) who does not share race neither with $i$ nor with $j$. These are given by
\begin{eqnarray}
\label{eq:pa}
P(A)&=&1-\underbrace{(1-p)}_{E(i,i')=0} \quad \underbrace{(1-p^2)^{n-2}}_{E(i,i'')=E(i'',i')=0} \quad \underbrace{(1-q^2)^{(r-1)n}}_{E(i,h)=E(h,i')=0}\\
\label{eq:pb} P(B)&=&1-\underbrace{(1-q)}_{E(i,j)=0} \quad \underbrace{(1-pq)^{2n-2}}_{E(i,i')=E(i',j)=0} \quad \underbrace{(1-q^2)^{(r-2)n}}_{E(i,h)=E(h,j)=0}
\end{eqnarray}
Plugging the values computed in equations \eqref{eq:pa} and \eqref{eq:pb} into \eqref{eq:expecteddiversitylong}, we can plot that function and observe that it grows very fast: after $q$ becomes positive, the diversity of a society quickly becomes approximately one. To understand the rapid increase in diversity, let us fix $p=1$ and let $q=1/n$. Then 
\begin{eqnarray}
P(B)&=&1-(1-q)^{2n-1}(1-q^2)^{(r-2)n}\\
&=&1-(1-q)^{rn-1}(1+q)^{(r-2)n}\\
&=&1-(1-\frac{1}{n})^{rn-1}(1+\frac{1}{n})^{(r-2)n}\\
&\underset{n \to \infty}{=} & 1- e^{-2} \quad \approx \quad 0.86
\end{eqnarray}

Substituting the value of $P(B)$ into \eqref{eq:expecteddiversitylong}, we obtain that $\EE[dv(S_{\text{long}})] \approx \frac{.86r}{.86r+.14}$, which is very close to 1 even when $r$ is small ($\EE[dv(S_{\text{long}})] \approx .92$ already for $r=2$), showing that the diversity of a society becomes 1 for very small values of $q$, in particular $q=1/n$. The intuition behind full diversity for the case of long marriages is that, once an agent obtains just one edge to any other race, he gains $\frac{n}{2}$ potential partners. Just one edge to a person of different race gives access to that person's complete race. 

Although we fixed $p=1$ to simplify the expressions of expected diversity, the rapid increase in diversity does not depend on each race having a complete graph. We also obtain a quick increase in diversity for many other values of $p$, as we discuss in Appendix \ref{app:robust}. When same-race agents are less interconnected among themselves, agents gain fewer connections once an interracial edge is created, but those fewer connections are relatively more valuable, because the agent had less potential partners available to him before.\footnote{This finding should not be confused with (and it is not implied by) two well-known properties of random graphs. The first one establishes that a giant connected component emerges in a random graph when $p=1/n$, whereas the graph becomes connected when $p=\log (n)/n$; for a review of these properties see \cite{albert2002}. The second result is that the property that a random graph has diameter 2 (maximal path length between nodes) has a sharp threshold at $p=(2\ln n/n)^{1/2}$ \citep{blum2017}. Result 1 is also similar to, but not implied by,  the {\it small world} property of simple random graphs \citep{watts1998}, where an average small path length occurs in a regular graph after rewiring a few initial edges.}

To further visualize the rapid increase in diversity we use simulations. We generate several random societies and observe how their average diversity change when they become more connected. We create ten thousand random societies, and increase the expected number of interracial edges by increasing the parameter $q$. In the simulations presented in the main text we fix $n=50$ and $p=1$.\footnote{We restrict to $n=50$ and ten thousand replications because of computational limitations. The results for other values of $p$ are similar and we describe them in Appendix \ref{app:robust}.} 

As predicted by our theoretical analysis, a small increase in the probability of interracial connections achieves perfect social integration in the case of long marriages.\footnote{Perfect social integration (diversity equals one) occurs around $q=1/n$, as we have discussed. The emergence of perfect integration is not a phase transition but rather a crossover phenomenon, i.e. diversity smoothly increases instead of discontinuously jumping at a specific point: see Figure B1 in Appendix \ref{app:robust}.}$\textcolor{blue}{^,}$\footnote{This result is particularly robust as it does not depend on our assumption that the marriages created are stable. Stability is not innocuous in our model, as we could consider other matching schemes that in fact are edge-monotonic.} For the cases with direct marriages, the increase in diversity is slower but still fast: an increase of $q$ from 0 to 0.1 increases diversity to 0.19 for $r=2$, and from 0 to 0.37 with $r=5$.\footnote{Empirical evidence strongly suggests that $q$ is very close to zero in real life. See footnote 13.} Figure \ref{fig:divplot} summarizes our main result, namely:
\begin{figure}[ht]
        \includegraphics[width=.9\textwidth]{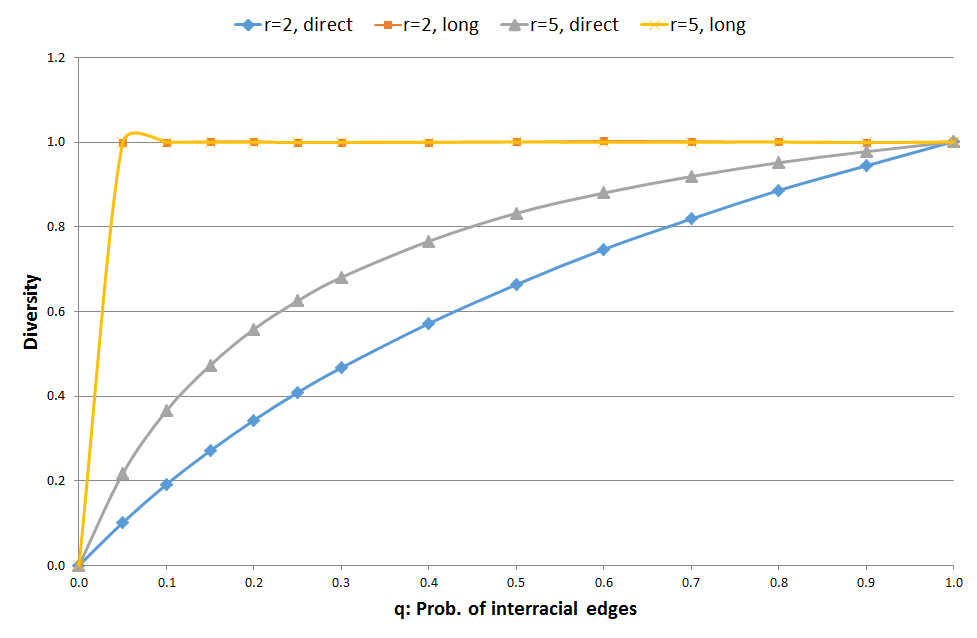}
					\label{plot:diversity}
					\caption{Average diversity of an Euclidean society for different values of $q$.}
					\begin{tablenotes}
							\item  \footnotesize The yellow and orange curves are indistinguishable in this plot because they are identical. Exact values and standard errors (which are in the order of 1.0e-04) are provided in Appendix \ref{app:tables}, as well as the corresponding graph for an assortative society, which is almost identical.
					\end{tablenotes}
					\label{fig:divplot}
\end{figure}

\begin{result}
\label{res:diversity}
Diversity is fully achieved with long marriages, even if the increase in interracial connections is arbitrarily small. 

With direct marriages, diversity is achieved partially, yet an increase in $q$ around $q=0$ yields an increase of a larger size in diversity.
\end{result}

We have showed that with either direct ($k=1$) or long ($k=2$) marriages diversity increases after the emergence of online dating, although at very different rates. An obvious question is whether online dating actually helps to create long marriages. We study the case of long marriages not because we expect that if a man meets a woman online, then that man will be able to date that woman's friends. Rather, we study it because it shows that when people meet their potential spouses via friends of friends ($k>2$), a few existing connections can quickly make a difference: recall that meeting through friends of friends is the most common way to meet a spouse both in the US and Germany (around one out of every three marriages start this way in both countries \citep{rosenfeld2012,potarca2017}).\footnote{\cite{ortega2018} finds the minimal number of interracial edges needed to guarantee that any two agents can marry for all values of $k$.} 

Our analysis shows that immediate social integration occurs for all values of $k\geq 2$. The mechanism we consider for those larger values of $k$ is that, once an interracial couple is created, it serves as a bridge between two different races. For example, if woman $a$ marries man $b$ of a different race, in the future it allows agent $a'$, an acquaintance of woman $a$, to meet agent $b'$, an acquaintance of man $b$, allowing $a'$ and $b'$ to marry. In summary, we expect that some marriages created by online dating will be between people who meet directly online, but some will be created as a consequence of those initial first marriages, and thus the increase in the diversity of societies will be somewhere in between the direct and the long marriage case.

Result \ref{res:diversity} implies that a few interracial links can lead to a significant increase in the racial integration of our societies, and leads to optimistic views on the role that dating platforms can play in modern civilizations. Our result is in sharp contrast to the one of \citet{schelling1969,schelling1971} in its seminal models of residential segregation, in which a society always becomes completely segregated. We pose this finding as the first testable hypothesis of our model.
\begin{hypo}
The number of interracial marriages increases after the popularization of online dating.
\end{hypo}

\subsection{Strength \& Size}
A second observation, less pronounced than the increase in diversity, is that strength is increasing in  $q$. We obtain this result by using simulations only, given that it seems impossible to obtain an analytical expression for the expected strength of a society.\footnote{Solving the expected average distance in a toy society with just one race, containing only one man and one woman, requires a long and complicated computation \href{https://mindyourdecisions.com/blog/2016/07/03/distance-between-two-random-points-in-a-square-sunday-puzzle/}{``Distance between two random points in a square''}, {\it Mind your Decisions}, 3/6/2016.} Figure \ref{fig:stplot} presents the average strength of the marriages obtained in ten thousand simulations with $n=50$ and $p=1$.
\begin{figure}[ht]
        \includegraphics[width=.9\textwidth]{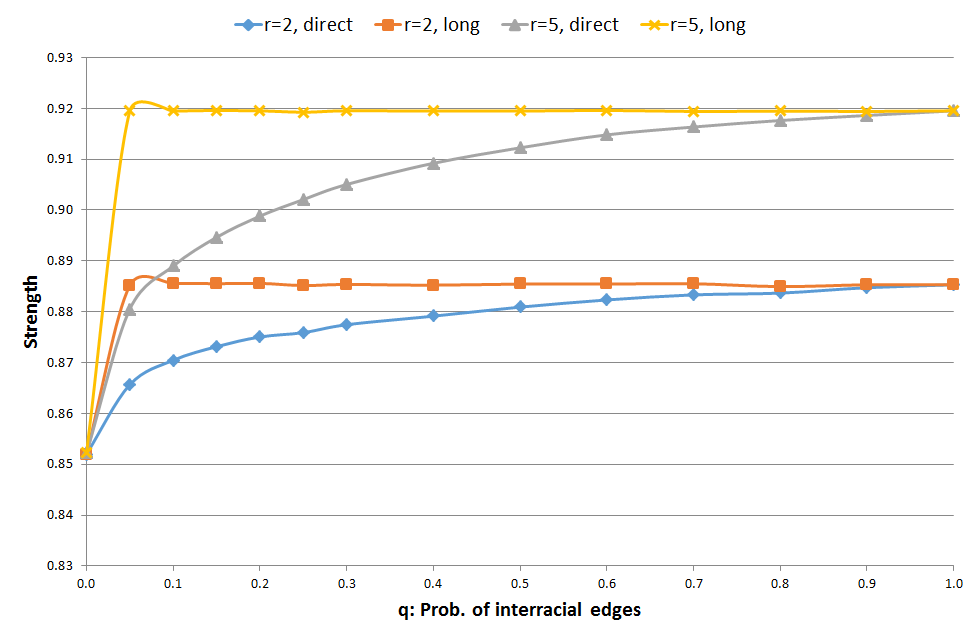}
				\caption{Average strength of an Euclidean society for different values of $q$.}
				\begin{tablenotes}
						\item \footnotesize	Exact values and standard errors (which are in the order of 1.0e-04) provided in Appendix \ref{app:tables}, as well as the corresponding graph for an assortative society, which is very similar.
					\end{tablenotes}
				\label{fig:stplot}
\end{figure}

The intuition behind this observation is that agents have more partner choices in a more connected society. Although this does not mean that every agent will marry a more desired partner, it does mean that the average agent will be paired with a better match. It is clear that, for all combinations of parameters (see Appendix \ref{app:robust} for further robustness checks), there is a consistent trend downwards in the average distance of partners after adding new interracial edges, and thus a consistent increase in the strength of the societies. We present this observation as our second result.
	
\begin{result}
\label{res:strength}
Strength increases after the number of interracial edges increases. The increase is faster with long marriages and with higher values of $r$.
\end{result}

Assuming that marriages between spouses who are further apart in terms of personality traits have a higher chance of divorcing because they are more susceptible to break up when new nodes are added to the society graph, we can reformulate the previous result as our second hypothesis.

\begin{hypo}
Marriages created in societies with online dating have a lower divorce rate.
\end{hypo}

Finally, with regards to size, we find that the number of married people also increases when $q$ increases. This observation, however, depends on $p<1$.\footnote{Using Hall's marriage theorem, \cite{erdos1964} find that in a simple random graph ($r=1$) the critical threshold for the existence of a {\it perfect matching} is $p=\log n/ n$, i.e. a marriage with size 1. Even when $p=q$, this critical threshold is only a lower bound for a society to have size 1. This is because there is no guarantee that the stable matching will in fact be a perfect one.} This increase is due to the fact that some agents do not know any available potential spouse who prefers them over other agents. Figure \ref{fig:szplot} presents the evolution of the average size of a society with $p=1/n$.
\begin{figure}[ht]
        \includegraphics[width=.9\textwidth]{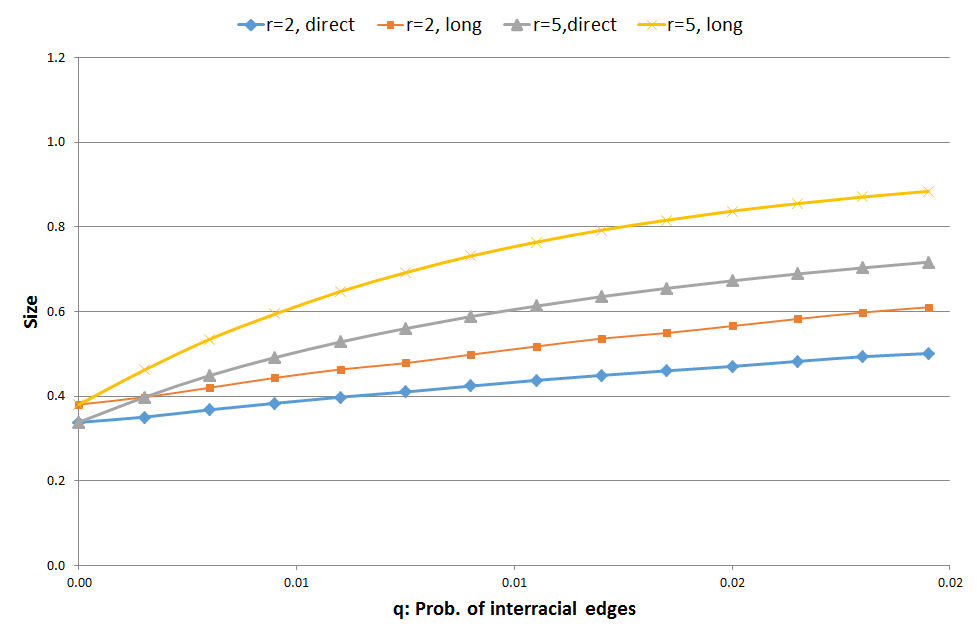}
				\caption{Average size of an Euclidean society for values of $q$ up to $p=1/n$.}
				\begin{tablenotes}
						\item \footnotesize	Exact values and standard errors (which are in the order of 1.0e-04) provided in Appendix \ref{app:tables}, as well as the corresponding graph for an assortative society, which is very similar.
					\end{tablenotes}
				\label{fig:szplot}
\end{figure}

The increase in the number of married people becomes even larger (and does not require $p<1$) whenever i) some races have more men than women, and vice versa,\footnote{See \cite{ahn2018} for empirical evidence on how gender imbalance affects cross-border marriage.} ii) agents become more picky and are only willing to marry an agent if he or she is sufficiently close to them in terms of personality traits, or iii) some agents are not searching for a relationship. All these scenarios yield the following result.
\begin{result}
\label{res:size}
Size increases after the number of interracial edges increases if either 	$p<1$, societies are unbalanced in their gender ratio, or some agents are deemed undesirable. The increase is faster with long marriages and with higher values of $r$.
\end{result}

The previous result provides us with a third and final testable hypothesis, namely:
\begin{hypo}
The number of married couples increases after the popularization of online dating.
\end{hypo}

\section{Hypotheses and Data}
\label{sec:data}
\subsection{Hypothesis 1: More Interracial Marriages}
\begin{figure}[t!]
\centering
	\includegraphics[width=.97\textwidth]{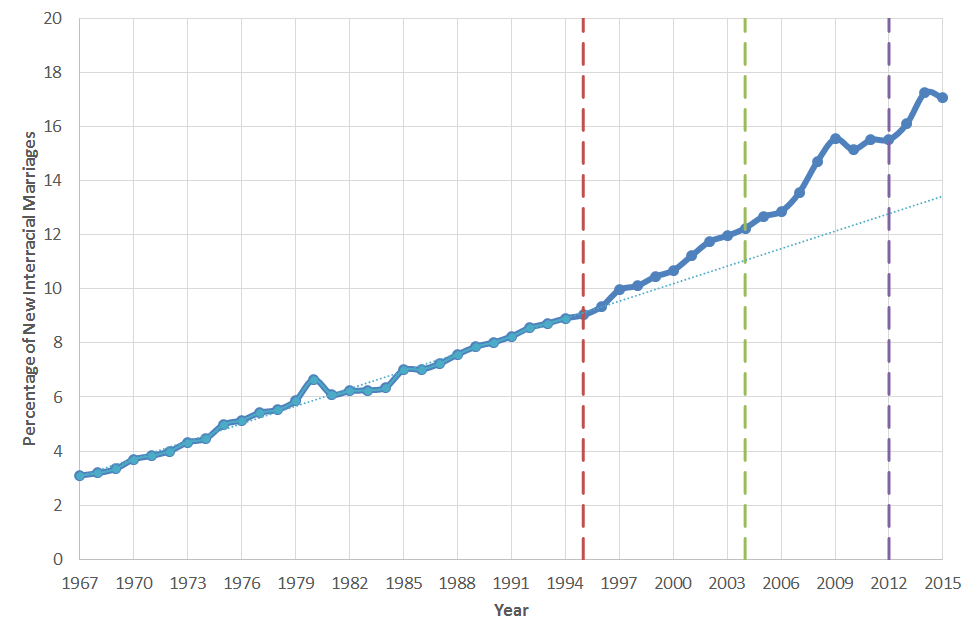}
	\caption[caption]{Percentage of interracial marriages among newlyweds in the US.}
	\label{fig:pew}
	\begin{tablenotes}[Source]
		\item \footnotesize Source: Pew Research Center analysis of 2008-2015 American Community Survey and the 1980, 1990 and 2000 decennial censuses (IPUMS). The red, green, and purple lines represent the creation of Match.com, OKCupid, and Tinder. The creation of Match.com roughly coincides with the popularization of broadband in the US and the 1996 Telecommunications Act. The blue line represents a linear prediction for 1996 -- 2015 using the data from 1967 to 1995.
\end{tablenotes}
\end{figure}	

What does the data reveal? Is our model consistent with observed demographic trends? We start with a preliminary observation before describing our empirical work in the next subsection. Figure \ref{fig:pew} presents the evolution of interracial marriages among newlyweds in the US from 1967 to 2015, based on the 2008-2015 American Community Survey and the 1980, 1990 and 2000 decennial censuses (IPUMS). In this Figure, interracial marriages include those between White, Black, Hispanic, Asian, American Indian or multiracial persons.\footnote{We are grateful to Gretchen Livingston from the Pew Research Center for providing us with the data. Data prior to 1980 are estimates. The methodology on how the data was collected is described in \cite{pew2017}.}

We observe that the number of interracial marriages has consistently increased in the last 50 years. However, it is intriguing that a few years  after the introduction of the first dating websites in 1995, like Match.com, the percentage of new marriages created by interracial couples increased. The increase becomes steeper around 2006, a couple of years after online dating became more popular: it is around this time when well-known platforms such as OKCupid emerged. During the 2000s, the percentage of new marriages that are interracial rose from 10.68\% to 15.54\%, a huge increase of nearly 5 percentage points, or 50\%. After the 2009 increase, the proportion of new interracial marriage jumps again in 2014 to 17.24\%, remaining above 17\% in 2015 too. Again, it is interesting that this increase occurs shortly after the creation of Tinder, considered the most popular online dating app.\footnote{Tinder, created in 2012, has approximately 50 million users and produces more than 12 million matches per day. See \href{https://www.nytimes.com/2014/10/30/fashion/tinder-the-fast-growing-dating-app-taps-an-age-old-truth.html}{``Tinder, the fast-growing dating app, taps an age-old truth''}, {\it New York Times}, 29/10/2014. The company claims that 36\% of Facebook users have had an account on their platform.} 

The increase in the share of new marriages that are interracial could be caused by the fact that the US population is in fact more interracial now than 20 years ago. However, the change in the population composition of the US cannot explain the huge increase in intermarriage that we observe, as we discuss in detail in Appendix \ref{app:population}. A simple way to observe this is to look at the growth of interracial marriages for Black Americans. Black Americans are the racial group whose rate of interracial marriage has increased the most, going from 5\% in 1980 to 18\% in 2015. However, the fraction of the US population that is Black has remained constant at around 12\% of the population during the last 40 years. Random marriage accounting for population change would then predict that the rate of interracial marriages would remain roughly constant, although in reality it has more than tripled in the last 35 years.

The correlation between the increase in the number of interracial marriages and the emergence of online dating is suggestive, but the rise of interracial marriage may be due to many other factors, or a combination of those. To precisely pin down the effect of online dating in this increase, we proceed as follows. 

\subsection{Empirical Test of Hypothesis 1}
We use the following strategy in order to rigorously test our prediction that online dating increases the number of interracial marriages. Our empirical setup exploits state variations in the development of broadband internet from 2000 to 2016, which we use as a proxy for online dating. There is little concern for reverse causality, which would imply that broadband developed faster in states where there was a higher number of interracial couples. Our dependent variable is a dummy showing whether a person's marriage is interracial. We use a variety of personal and state-level covariates in order to identify the relationship between online dating and interracial marriages as precisely as possible. Figure \ref{fig:scatter} displays a preview of the relationship between broadband development and interracial marriage by state.

\begin{figure}[ht!]
	\centering
	\includegraphics[width=.85\textwidth]{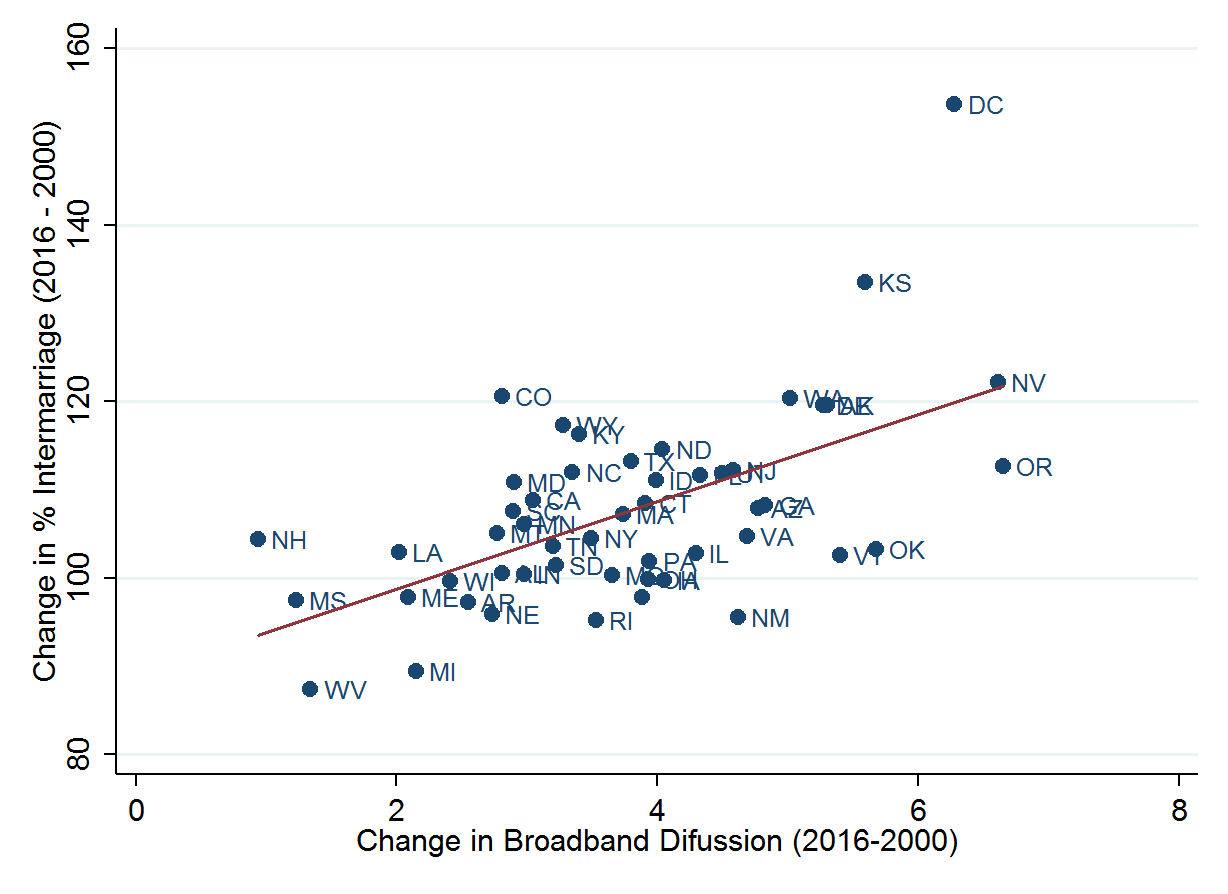}
	\caption{Change in \% of marriages that are interracial in US, by state.}
	\label{fig:scatter}
		\begin{tablenotes}[Source] \centering
				\item \footnotesize Source: FCC statistical reports on broadband development, US Census population estimates, and the American Community Survey (IPUMS) from 2000 to 2016.
					\end{tablenotes}
\end{figure}	

We use three main data sources for our analysis. All data concerning individuals is downloaded from IPUMS, and we restrict our analysis to married individuals only. Although the data is only on the individual level, it is possible to construct marriage relationships, by employing a matching procedure described at IPUMS website. As additional control variables, we employ education, age, and total income,\footnote{One might worry about endogeneity coming from income, as the marriage decision might affect earnings. Excluding income as explanatory variable leaves the coefficients virtually unchanged. As additional robustness check, we estimate a similar model at state level in Appendix \ref{app:estimation}.} as these are likely to affect the marriage decision. 

We construct the broadband data using information from reports by the Federal Communications Commission (FCC), which is the regulatory authority in the United States responsible for communication technology. Following \cite{bellou2015},\footnote{She uses a similar specification to examine the role of internet diffusion in the creation of new marriages. Our dataset is described in detail in Appendix \ref{app:estimation}.} we use the number of residential high-speed internet lines per 100 people as our explanatory variable. Data is available for the years 2000 to 2016. However, we have to discard Hawaii from our analysis, as observations are missing up to 2005.

We download additional state controls from the Current Population Survey. Following Bellou's work, we include variables like the ratio of the male divided by female population within a state, age bins and the ratio of non-white people in a state. This last explanatory variable is especially important in our context of interracial marriages. 

We estimate the following reduced form equation by a linear probability model: 
\begin{equation}
\label{eq:emp_model}
\text{Inter}_{ist} = \alpha + \beta \, \text{Broadband}_{st} + \gamma_1 X_{ist} + \gamma_2 Z_{st} + FE_s + FE_t + \epsilon_{ist}
\end{equation}

where $\text{Inter}_{ist}$ is one if a person is in an interracial marriage and 0 otherwise. The indices relate to person $i$, living in state $s$ at time $t$. We are mostly interested in the coefficient $\beta$, as it captures the propensity of online dating. The values in $X$ are covariates relating directly to the person, while $Z$ represents time varying state variables. We additionally include state- and year fixed effects, and cluster the standard errors $\epsilon_{ist}$ at the state-year level. Our rich battery of control variables enables us to clearly identify the relationship between interracial marriages and broadband internet, which can be seen as an instrument for online dating. As marriages take a while to form, we include the broadband variable with a 3 year lag based on empirical evidence \citep{rosenfeld2017}.\footnote{In Appendix \ref{app:estimation} we follow a different strategy. We construct shares of interracial marriages per state and year and estimate this with panel methods. The advantage is that the dependent variable is continuous rather than dichotomous, however we cannot use individual controls and introduce standard errors via aggregation. These standard errors should be negligible given the amount of observations we have available. The state-year level specification also generates significant coefficients with the expected signs, confirming our results. }

\begin{table}[ht]
\scalebox{0.9}{
\def\sym#1{\ifmmode^{#1}\else\(^{#1}\)\fi}
\begin{tabular}{l*{5}{c}}
\hline\hline
        
            &\multicolumn{5}{c}{Interracial Marriage}\\
						&\multicolumn{1}{c}{(1)}&\multicolumn{1}{c}{(2)}&\multicolumn{1}{c}{(3)}&\multicolumn{1}{c}{(4)}&\multicolumn{1}{c}{(5)}\\
\hline
Broadband (-3)  &  .00071\sym{***}&    .00058\sym{***}&    .00065\sym{***}&    .00053\sym{***}&                     \\
            & {\small (.000057)}         & {\small(.000066)}         & {\small (.000056)}         & {\small(.000065)}         &                     \\
[1em]
Broadband   &                     &                     &                     &                     &    .00021\sym{***}\\
            &                     &                     &                     &                     & {\small (0.000059)}         \\
[1em]
Age   &                     &                     &    -.0033\sym{***}&    -.0033\sym{***}&    -.0033\sym{***}\\
            &                     &                     & {\small (.000032)}         & {\small (.00032)}         & {\small (.000032)}         \\
[1em]
Education        &                     &                     &    -.0024\sym{***}&    -.0024\sym{***}&    -.0024\sym{***}\\
            &                     &                     &  {\small (.00023)}         &  {\small (.00023)}         &  {\small (.00023)}         \\
[1em]
Log Income      &                     &                     &   -.0056\sym{***}&   -.0056\sym{***}&   -.0056\sym{***}\\
            &                     &                     &  {\small (.00015)}         &  {\small (.00015)}         &  {\small (.00015)}         \\
[1em]
State controls & & x & & x & x \\
\hline
\(N\)       &    17,284,584         &    17,284,584         &    17,284,584         &    17,284,584         &    17,284,584         \\
Adj. \(R^{2}\)&       0.021         &       0.021         &       0.048         &       0.048         &       0.048         \\
\hline\hline
\multicolumn{5}{l}{\footnotesize Standard errors are in parentheses and clustered at state-year level.}\\
\multicolumn{5}{l}{\footnotesize All regressions include state and year dummies.}\\
\multicolumn{5}{l}{\footnotesize \sym{*} \(p<0.05\), \sym{**} \(p<0.01\), \sym{***} \(p<0.001\)}\\
\end{tabular}}
\caption{Effect of broadband diffusion on interracial marriage.}
\label{tab:estim}
\end{table}

The first column in Table \ref{tab:estim} states that one additional line of broadband internet 3 years ago affects the probability of being in an interracial marriage by 0.07\%. The coefficient is positive, as predicted by our theoretical model. In column (2) we include controls at the state level and find that the relationship between interracial marriages and broadband remains significantly positive. This continues to be true when including the individual covariates, all of which decrease the probability of a marriage being interracial. Perhaps surprisingly, education enters negatively. A potential underlying reason might be that education leads to more segregated friendship circles, a conjecture worth being explored in subsequent work. 

Column (4) is now the specification outlined in \eqref{eq:emp_model}. Even with all controls, the effect of broadband penetration on interracial marriages is highly significant and positive. This result suggests a causal relationship in the sense described by our model. As additional evidence for this claim, we see that once we replace the lagged broadband with its contemporaneous counterpart, the coefficient declines in size, which means that the state of broadband 3 years ago has a bigger effect on interracial marriages as compared to broadband today. This is because it takes time for marriages to form. 

Overall, the work we have presented here, jointly with robustness checks described in Appendix \ref{app:estimation}, suggests that there is empirical support for our hypothesis of online dating leading to more interracial marriages. 

Furthermore, the work of \cite{thomas2018}, released shortly after we made the first version of our paper available online, has provided further evidence of the role of online dating in the creation of new interracial marriages. Using a self-collected dataset representative of the US population (known as ``how couples met and stayed together'' or HCMST), Thomas finds that couples who met online were more likely to be interracial, even after controlling for the racial composition of their locations and confounding factors. In particular, after analyzing information about 3,036 American couples, he finds that couples who met online since 1996 are 6 to 7 percent more often interracial than couples who met purely offline. His finding, using different methods and data, is similar to ours and provides further support for Hypothesis 1. His dataset is freely available online for replication purposes.

\subsection{Hypothesis 2 \& 3: More and Better Marriages}
With regards to Hypothesis 2 and 3, which establish the creation of better and more marriages, respectively, we do not provide novel empirical work but we survey existing research from different disciplines.

There are two articles which have focused on whether relationships created online last longer than those created elsewhere. The first one is \cite{cacioppo2013}. They find that marriages created online were less likely to break up and exhibited a higher marital satisfaction, using a sample of 19,131 Americans who married between 2005 and 2012. They write: {\it ``Meeting a spouse on-line is on average associated with slightly higher marital satisfaction and lower rates of marital break-up than meeting a spouse through traditional off-line venues"}. The second one is \cite{rosenfeld2017}. Analyzing the HCMST dataset from 2009 to 2015, he finds no difference in the duration of marriages that start online and offline. Besides their methodological differences, what it is clear is that both papers find that marriages created online last at least as long as those created elsewhere, disproving the common popular belief that online relationships are only casual and of lower quality (see footnote 15). This finding aligns with Hypothesis 2 of our model.

With regards to Hypothesis 3, which states that the advent of online dating leads to a higher number of marriages, there is in fact empirical evidence supporting it. \cite{bellou2015} examines the role that internet penetration (in the form of broadband deployment) has had in the number of White and Black young adults who decide to marry. She uses data from the Current Population Survey and the FCC from 2000 to 2010. She finds that wider internet availability has indeed {\it caused} more interracial marriages among people between 21 and 30 years old. In particular, she finds that marriage rates are currently higher by 13\% to 33\% from what they might have been if the internet had not been available, despite a pre-existing downward trend in the propensity to marry among young adults.

\section{Final Remarks}
\label{sec:conclusion}

\subsection{Limitations of our Model}
Our model does not explain three observed characteristics of interracial marriages. First, it does not explain why interracial marriages are more likely to end up in divorce \citep{bratter2008,zhang2009}. Second, it does not explain why some intraracial marriages from a particular race last longer than intraracial marriages from another race (e.g. \citealp{stevenson2007} document that Blacks who divorce spend more time in their marriage than their White counterparts). And third, our model does not explain why interracial marriage between specific combinations of race and gender are more common than others (marriage between White men and Asian women is much more common than marriage among Asian men and White women; similarly marriage between Black men and White women is much more common than marriage between Black women and White men, see \citealp{chiappori2016black}.). A theoretical model that can account for all those stylized facts is still missing (see \citealp{fryer2007} for a discussion of how well existing models of marriage explains observed interracial marriage trends).

\subsection{Further Applications}
The theoretical model we present discusses a general matching problem under network constraints, and hence it can be useful to study other social phenomena besides interracial marriage. Furthermore, the role of connecting highly clustered groups is also not only linked to online dating. Another example is the European student exchange program ``Erasmus'', which helped more than 3 million students and over 350 thousand academics and staff members to spend time at a University abroad.\footnote{\href{http://ec.europa.eu/dgs/education_culture/repository/education/library/statistics/ay-12-13/facts-figures_en.pdf}{``ERASMUS: Facts, figures and trends.''}, {\it European Commission}, 10/6/2014. Interestingly, \cite{parey2011} find that participating in ERASMUS increases the probability of working abroad by 15 percentage points. Their data suggests that a large fraction of this effect comes from marrying a foreign partner.} Although it would be interesting to test our model in these and other scenarios, we leave this task for further research.

\subsection{Conclusion}
We introduce a theoretical model to analyze the complex process of deciding whom to marry in the times of online dating. Our model is admittedly simple and fails to capture many of the complex features of romance in social networks, like love. However, in our view, the simplicity of our model is its main strength. It generates strong predictions with a simple structure. The main one is that the diversity of societies, measured by the number of interracial marriages in it, increases after the introduction of online dating. Not only is this prediction consistent with demographic trends, but an empirical analysis of interracial marriages within each US state suggests that online dating is indeed partially responsible for the observed increase in interracial marriage. And if that is the case, in words of the \cite{mit}: {\it ``the model implies that this change is ongoing. That’s a profound revelation. These changes are set to continue, and to benefit society as result''}.

Simple models are great tools for conveying an idea. Schelling's segregation model clearly does not capture many important components of how people decide where to live. It could have been enhanced by introducing thousands of parameters. Yet, it has broadened our understanding of racial segregation, and has been widely influential: according to Google Scholar, it has been quoted 3,258 times by articles in a variety of field ranging from sociology to mathematics. It has provided us with a way to think about an ubiquitous phenomenon.

Our model is a modest attempt that goes in the same direction.

\newpage
\bibliographystyle{ecta}

\newpage 

\appendix
\begin{center}
	Appendices (for online publication only).
\end{center}

\section{Appendix A: Simulation Results}
\label{app:tables}

\section{Appendix B: Robustness Checks}
\label{app:robust}

\section{Appendix C: Interracial Marriages and Population Composition}
\label{app:population}

\section{Appendix D: Regression Analysis}
\label{app:estimation}
\end{document}